\documentclass[a4paper]{amsart}

\usepackage[T1]{fontenc}
\usepackage{lmodern}
\usepackage{microtype}
\usepackage{amsmath}
\usepackage{mathtools}
\usepackage{amssymb}
\usepackage{amsthm}
\usepackage{mathrsfs}  
\usepackage{xcolor}
\usepackage{cite}
\usepackage{comment}
\usepackage{forest}
\usepackage{hyperref}
\usepackage{cleveref}
\usepackage{nicematrix}
\usepackage{tikz}
\usepackage{tikz-cd}
\usetikzlibrary{matrix} 
\usepackage{xspace}
\usepackage{enumitem}
\setlist[itemize]{leftmargin=*}
\setlist[enumerate]{leftmargin=*}

\DeclareRobustCommand{\SkipTocEntry}[5]{}

\usepackage{calc}
\usepackage{xfrac}
\usepackage{braket}
\renewcommand{\colon}{\vcentcolon}
\newcommand{\coloncolon}{\mathrel{\vcentcolon\vcentcolon}}
\newcommand{\colonequals}{\mathrel{\vcentcolon=}}
\newcommand{\defeq}{\colonequals}

\newcommand{\GHZ}{\psi_{_\textsf{GHZ}}}

\newcommand{\ctx}{\Sigma}
\newcommand{\Dist}{\mathcal{D}}
\newcommand{\Pow}{\mathscr{P}}
\newcommand{\supp}{\mathsf{supp}}
\newcommand{\op}{\mathsf{op}}
\newcommand{\CSet}{\mathbf{Set}}

\theoremstyle{plain}
\newtheorem{theorem}{Theorem}[section]    
\newtheorem{lemma}[theorem]{Lemma}
\newtheorem{proposition}[theorem]{Proposition}

\theoremstyle{definition}
\newtheorem{definition}[theorem]{Definition}

\theoremstyle{remark}
\newtheorem{remark}[theorem]{Remark}

\newcommand{\ie}{\text{i.e.\@}\xspace}
\newcommand{\eg}{\text{e.g.\@}\xspace}

\newcommand{\ZZ}{\mathbb{Z}}
\DeclareMathOperator{\rank}{rank}
\DeclareMathOperator{\dom}{dom}
\DeclareMathOperator{\im}{Im}
\DeclareMathOperator{\Th}{Th}

\newcommand{\MP}{\mathsf{MP}}
\newcommand{\ES}{\mathscr{E}} 
\newcommand{\II}{\mathcal{I}} 
\newcommand{\gs}{g} 
\newcommand{\ls}{h} 
\newcommand{\bell}[1][\II]{B_#1}
\newcommand{\branching}[3]{\left(#1,#2,#3\right)} 
\newcommand{\multiplicative}[3]{\left\langle #1,#2,#3\right\rangle}
\DeclareMathOperator{\Branches}{Br}
\newcommand{\FF}{\mathcal{F}}

\newcommand{\OR}{\vee}

\newcommand{\AvN}{\mathsf{AvN}}
\newcommand{\CSC}{\mathsf{CSC}}

 \usepackage{soul}
 \usepackage{todonotes}

\forestset{
  dense/.style={
    for tree={
      l*=0.8,
     l sep*=0.8,      
      s sep*=0.8,      
      inner sep=3pt  
    },
     baseline,}}

\usepackage{orcidlink}
\let\oldorcidlink\orcidlink
\renewcommand{\orcidlink}[1]{
  \begingroup
  \hypersetup{pdfborder={0 0 0}}
  \oldorcidlink{#1}
  \endgroup
}

\begin{document}

\title{Algebraic paradoxes in adaptive quantum computation}

\author[S.~Abramsky]{Samson Abramsky}
\address{\parbox{\linewidth}{
            Samson Abramsky\,\orcidlink{0000-0003-3921-6637}
         \\ Department of Computer Science, University College London
         \\ London, United Kingdom
         }}
\email{s.abramsky@ucl.ac.uk}

\author[R.~S.~Barbosa]{Rui Soares Barbosa}
\address{\parbox{\linewidth}{
            Rui Soares Barbosa\,\orcidlink{0000-0002-0465-8518}
         \\ International Iberian Nanotechnology Laboratory
         \\ Braga, Portugal
         }}
\email{rui.soaresbarbosa@inl.int}

\author[C.~Constantin]{Carmen Constantin}
\address{\parbox{\linewidth}{
            Carmen Constantin\,\orcidlink{0000-0003-4508-9312}
         \\ Department of Computer Science, University College London
         \\ London, United Kingdom
         }}
\email{c.constantin@ucl.ac.uk}

\author[M.~Karvonen]{Martti Karvonen}
\address{\parbox{\linewidth}{
            Martti Karvonen\,\orcidlink{0000-0002-8919-343X}
         \\ Department of Computer Science, University College London
         \\ London, United Kingdom
         \\[2pt] \emph{From August 2026:} Department of Computer Science,
            University of Bath
         \\ Bath, United Kingdom
         }}
\email{martti.karvonen@ucl.ac.uk}

\begin{abstract}
Measurement-based quantum computation (MBQC) is a universal model of quantum computation whose full power requires adaptivity.
Contextuality is known to power quantum advantage in MBQC, yet it has resisted algebraic analysis in the adaptive setting.

We show that if an adaptive $\mathbb{Z}_2$-linear measurement-based quantum computing protocol deterministically computes a non-affine Boolean function, then the underlying quantum resource satisfies an inconsistent set of linear equations.
This witnesses an algebraic form of strong contextuality generalising Mermin's All-versus-Nothing arguments.
Such algebraic contextuality can be detected cohomologically, resolving an open question posed by Raussendorf,
who had established cohomological witnesses of contextuality for non-adaptive protocols, but left the adaptive case open.
We prove this result constructively:
we model adaptive measurement protocols as ordinary measurements on a larger 
scenario of tree-like measurements,
and explicitly build the inconsistent 
equations inductively.
\end{abstract}

\vspace*{-.37cm}
\maketitle
\vspace{-.35cm}
\tableofcontents
\clearpage

\section{Introduction}\label{sec:introduction}
If quantum computers are to outperform their classical counterparts,
their additional power must stem from observable behaviour that lies beyond the reach of classical systems.
Contextuality -- the impossibility of reconciling quantum measurement statistics with pre-existing outcome values -- stands out as such a provably nonclassical phenomenon \cite{kochen1967problem,bell1966problem} 
and has been shown to enable quantum advantage across multiple models of computation and information-processing tasks \cite{raussendorf:contextualityinmbqc,howard2014contextuality,abramsky2017contextual,kirby2020classical,bravyi2018quantum}.
Closer to the metal, contextuality is being explored for benchmarking quantum hardware \cite{kumar2025quantumclassical}, as the constraints it imposes enable error detection beyond standard randomised benchmarking.
At the heart of contextuality lies a logical tension -- a ``paradox'':
each empirically-accessible set of compatible measurements (\emph{context}) provides a consistent classical snapshot, yet attempting to combine these into a single global description  leads to a contradiction.
This local-to-global obstruction is neatly captured in the language of sheaf theory~\cite{ab}.

Particularly elegant are purely algebraic witnesses of contextuality -- known as All-versus-Nothing (AvN) arguments -- in which a system of linear equations establishes such a logical contradiction: quantum measurement outcomes satisfy each equation locally, yet no global classical value assignment can satisfy them all.
Classic examples of such algebraic contextuality include the GHZ paradox \cite{greenberger1989going} and the Peres--Mermin magic square \cite{mermin1990simple,peres1990incompatible}, 
closely related to the parity-based games studied by Arkhipov \cite{arkhipov2012extending} and more generally quantum solutions to linear constraint systems in non-local games \cite{cleve2014characterization,slofstra2024operator}. Their rigid algebraic form makes such arguments particularly tractable: detecting inconsistency reduces to linear algebra rather than general constraint satisfaction, an advantage that becomes substantial at scale, and all the more relevant given growing interest in many-body contextuality~\cite{hart2025manybody,zhao2025scalable}.

Measurement-based quantum computation (MBQC)~\cite{raussendorf2001one,briegel2009mbqc}
is a standard universal model of quantum computation, equivalent in power to the  circuit model and actively pursued in optical and ion-trap architectures \cite{walther2005experimental,prevedel2007highspeed,lanyon2013mbqc}.\footnote{See the recent comparative survey \cite{larasatichoi2026circuitvsmbqc}, which discusses growing interest in this paradigm for hardware implementations.}
In this paradigm, a classical control computer orchestrates the computation  by performing successive local measurements on a multi-qubit resource state and processing the observed outcomes, treating the quantum system as a black box that responds to measurement queries.
This black-box perspective makes MBQC a particularly natural framework for analysing quantum resources through an operational lens, and it is indeed standard in the study of contextuality in MBQC~\cite{anders2009computational,raussendorf:contextualityinmbqc}.

Crucially in MBQC, the choice of later measurement settings may depend on prior outcomes, a feature known as \emph{adaptivity}.
This feed-forward capability is essential for quantum universality.
Yet many foundational results on MBQC have been confined to the non-adaptive setting.
The limitation runs deep: key algebraic and topological tools for analysing contextuality -- including AvN arguments and cohomological methods -- have hitherto only been applied in non-adaptive, ``flat'' measurement scenarios.
Raussendorf~\cite{raussendorf:cohomologicalframework,raussendorf:paradoxestowork} articulated an instance of this gap as an explicit open question: can cohomological witnesses for contextuality in MBQC be extended to adaptive protocols? This matters beyond MBQC: quoting~\cite{raussendorf:roleofcohomology}, `a web of cohomological facts relates quantum error correction, measurement-based quantum computation, symmetry protected topological order, and contextuality'. Understanding how adaptivity fits into this web is therefore of broad foundational importance.

We address a related but stronger question:
can the contextuality required for adaptive MBQC always be witnessed by an explicit algebraic contradiction -- an AvN argument?
The answer is yes.

To state our contribution precisely, we consider
$\mathbb{Z}_2$-linear MBQC,
where the classical control computer, including pre- and post-processing and adaptive control, is restricted to linear operations modulo two.
Raussendorf~\cite{raussendorf:contextualityinmbqc}
showed that if such a (possibly adaptive) $\ZZ_2$-MBQC protocol computes a nonaffine Boolean function deterministically, then the protocol must exploit strong contextuality in its resource state. 
Our main contribution strengthens Raussendorf's result by showing that the requisite strong contextuality is algebraic in nature: it can always be witnessed via an AvN argument.
Our proof is constructive: given an adaptive protocol, we explicitly build the inconsistent system of equations from its tree structure, rather than inferring its existence from properties of the resource state. This also implies the contextuality is cohomologically nontrivial, thus resolving Raussendorf's open question~\cite{raussendorf:cohomologicalframework,raussendorf:paradoxestowork}: cohomological witnesses for contextuality do extend to adaptive protocols.
We establish this in both the group-cohomological and \v{C}ech-cohomological frameworks for contextuality.

A key technical obstacle is that existing tools, including the very definition of AvN arguments \cite{abramskyetal:contextualitycohomologyparadox,abramsky2017complete,GogiosoZeng:AvNs}
and cohomological witnesses for contextuality \cite{abramskyetal:cohomology-of-contextuality,abramskyetal:contextualitycohomologyparadox,raussendorf:paradoxestowork,raussendorf:cohomologicalframework},
are formulated only for  ``flat'' measurement scenarios, making it unclear how to even state, let alone prove, adaptive analogues.
Rather than developing an entirely new theory for adaptive protocols, we employ a systematic ``flattening'' construction that models adaptive protocols as non-adaptive protocols on appropriately enlarged measurement scenarios.
Drawing on techniques from~\cite{abramskyetal:comonadicview},  where adaptivity is captured via a (co)Kleisli construction,
this perspective allows us to apply the full arsenal of existing tools directly in the adaptive setting, while making the inductive structure of the AvN arguments explicit; see also \cite{searle2024dphil,searle2024qpl_mapping} for similar flattening constructions in related settings.
We emphasise that this ``flattening'' does not eliminate adaptivity: measurements in the enlarged scenario are themselves adaptive operations on the original scenario, and would be implemented as such in practice.
The point is not to reduce adaptive to non-adaptive computation, but to bring mathematical tools to bear on the adaptive case.
Beyond the immediate result, our flattening construction provides a general pathway for extending results from the non-adaptive to the adaptive setting, enabling further study of adaptivity in MBQC.

\subsection*{An AvN argument: the GHZ paradox}
To illustrate AvN arguments and their relationship to measurement-based computation, we start by presenting a canonical example.

Consider a scenario in which a physical system is distributed over three sites.
At each of these, one may perform a local measurement and record the observed outcome.
There are two measurement settings available at each site,
corresponding to two measurement procedures one may perform locally on the system.
We label these $x_0, x_1$ for the first site, $y_0, y_1$ for the second, and $z_0, z_1$ for the third.
Each of these measurements is dichotomic, yielding an outcome in $\mathbb{Z}_2$.
Importantly, we view $\mathbb{Z}_2$ as a ring,
not merely as the bare set $\{0,1\}$; this algebraic structure is essential for expressing the kind of argument that follows.

In each run of the experiment, a measurement is performed at each site and the corresponding outcome is recorded.
Over many runs with varying choices of measurement settings,
one collects measurement statistics that characterise the observable behaviour.
These data constitute an \emph{empirical model}:
for each choice of measurement settings at the three sites -- a \emph{context} -- 
it specifies a probability distribution over the joint outcomes.
Crucially, we treat the system as a black box that responds to measurement queries, without presupposing any particular physical theory by which outcomes are produced: the model captures only what is observable.
\Cref{tab:GHZ} shows one particular empirical model for this scenario:
four contexts are shown; for the remaining four contexts, the joint outcomes are uniformly distributed.
Observe that the probability distributions over different contexts agree on their marginals: the choice of measurement setting at one site does not affect the probability distribution over outcomes at the other sites. This property is known as \emph{no-signalling} or \emph{no-disturbance}.

\begin{table}
\caption{The GHZ empirical model on $\bell[3]$}\label{tab:GHZ}
\begin{tabular}{c|cccccccc}
         & 000 & 001 & 010 & 011 & 100 & 101 & 110 & 111 \\ \hline
   $x_0 \, y_0 \, z_0$ & \sfrac{1}{4} &      0       &      0       & \sfrac{1}{4} &      0       & \sfrac{1}{4} & \sfrac{1}{4} &      0       \\
   $x_0 \, y_1 \, z_1$ &     0        & \sfrac{1}{4} & \sfrac{1}{4} &      0       & \sfrac{1}{4} &      0       &      0       & \sfrac{1}{4} \\
   $x_1 \, y_0 \, z_1$ &     0        & \sfrac{1}{4} & \sfrac{1}{4} &      0       & \sfrac{1}{4} &      0       &      0       & \sfrac{1}{4} \\
   $x_1 \, y_1 \, z_0$ &     0        & \sfrac{1}{4} & \sfrac{1}{4} &      0       & \sfrac{1}{4} &      0       &      0       & \sfrac{1}{4} \\
\end{tabular}
\end{table}

Classically, one thinks of measurements as revealing pre-existing properties of the physical system under observation.
More precisely, on each fully specified state (even if empirically inaccessible, \ie a hidden variable),
we expect a function $\gs$ that assigns a definite outcome to each measurement, independently of what other measurements are performed at the other sites;
the observable behaviour then arises from a probability distribution over such deterministic states.
The no-signalling property lends plausibility to this view.

However, examining the empirical model reveals a striking pattern.
In each of the four contexts shown, the observed joint outcomes always satisfy a  system of linear equations over $\mathbb{Z}_2$:
\newcommand{\fitX}[1]{\makebox[\widthof{$x_0$}][c]{$#1$}}
\begin{equation}\label{eq:ghz}
\begin{aligned}
    \fitX{x_0} + \fitX{y_0} + \fitX{z_0} &= 0 \qquad & 
    \fitX{x_0} + \fitX{y_1} + \fitX{z_1} &= 1 \\
    \fitX{x_1} + \fitX{y_0} + \fitX{z_1} &= 1 \qquad&
    \fitX{x_1} + \fitX{y_1} + \fitX{z_0} &= 1
\end{aligned}
\end{equation}
This system of equations is readily seen to be inconsistent:
adding all four equations, each measurement appears twice on the left-hand side and thus cancels, yielding $0=1$, a contradiction.
Therefore, no assignment $\gs \colon \{x_0,x_1,y_0,y_1,z_0,z_1\} \to \mathbb{Z}_2$ of definite values to the six measurements is consistent with the observed outcomes across the four contexts.
This establishes that the empirical model is \emph{strongly contextual} 
(see \Cref{sec:background} for a formal definition),
and algebraic contradictions of this form are known as \emph{All-versus-Nothing} (AvN) arguments \cite{abramskyetal:contextualitycohomologyparadox}.

Remarkably, quantum systems exhibit precisely such counterintuitive behaviour.
The empirical model in \Cref{tab:GHZ} can be realised using a three-qubit system initialised in the GHZ state
$\ket{\GHZ} = \sfrac{1}{\sqrt{2}}\left(\ket{000}+\ket{111}\right)$.
The measurements $x_0, x_1$ (resp.\@ $y_0, y_1$ and $z_0, z_1$) at each site
correspond to measuring the corresponding qubit in the $X = \{\ket{+},\ket{-}\}$ or $Y = \{\ket{+i},\ket{-i}\}$ bases, 
with the outcomes labelled $0$ and $1$.\footnote{The outcomes $0$ and $1$ correspond to the $+1$ and $-1$ eigenvalues of the respective Pauli operators via the group homomorphism $r \mapsto (-1)^r$.}

But we stress that such details about the quantum realisation are not essential for what follows: the key takeaway is that this seemingly paradoxical empirical model is physically realisable.
Throughout, we work at the level of observable behaviour (the statistics of measurement outcomes and the algebraic constraints they satisfy) rather than the Hilbert-space formalism, treating the quantum system as a black box.
This operational, theory-independent approach lays bare the ingredients driving computational advantage, independently of any particular physical model.

\subsection*{Computing with AvN paradoxes: the Anders--Browne construction}
Anders and Browne \cite{anders2009computational} showed how the GHZ empirical model can be used as a computational resource to enhance the power of a limited classical computer.
Specifically, they showed that a parity computer, restricted to performing affine operations over $\mathbb{Z}_2$, can compute non-affine Boolean functions when given access to this empirical model.
Such access takes the form of \emph{measurement-based computation}: the classical controller queries the model by choosing measurement settings and obtaining corresponding outcomes.

We present the construction for the (non-affine) OR function that computes the disjunction of two bits,
$\OR \colon \ZZ_2^2 \to \ZZ_2 \coloncolon (a,b) \mapsto a \OR b = a + b + a \cdot b$.
Given input bits $a, b \in \ZZ_2$, the parity control computer first applies the linear preprocessing function
$Q \colon \ZZ_2^2 \to \ZZ_2^3 \coloncolon  (a,b) \mapsto (a, b, a + b)$.
This encodes the input into a three-bit vector that determines the measurement settings for the three sites:
the computer queries the resource by measuring the context $\{x_a, y_b,  z_{a+b}\}$.
The three measurement outcomes are then post-processed by a linear function $Z \colon \ZZ_2^3 \to \ZZ_2$ that sums them, producing a single output bit.

The fact that this protocol \emph{deterministically} computes the OR function
follows from the linear equations in \cref{eq:ghz}, which can be condensed as
\begin{equation}\label{eq:ghz_condensed}
    x_a + y_b + z_{a+b} =  a \OR b \qquad (a,b \in \ZZ_2).
\end{equation}
Note that the empirical model can be regarded as a probabilistic function $e \colon \ZZ_2^3 \to \Dist(\ZZ_2^3)$, where $\Dist$ is the probability distribution monad.
The action of this protocol is then the Kleisli composition $Z \bullet e \bullet Q$,
which coincides with $\eta \circ \OR$, where $\eta$ is the unit of the monad $\Dist$ embedding deterministic values as point distributions.

One can use this construction as a building block to compute arbitrary Boolean functions,
as OR together with affine operations is functionally complete.\footnote{In fact,~\cite{anders2009computational} shows that access to polynomially-many GHZ states promotes the parity computer from (a subclass of) the complexity class $\oplus L$ to $P$.}
Each OR gate in a logical circuit is then implemented by this construction.
Achieving this in general requires \emph{adaptivity} or \emph{feed-forward} controlled by the parity computer: 
later measurement settings must depend (affinely) on prior measurement outcomes, allowing one to compose OR gates sequentially.

The GHZ algebraic paradox underlies the computation of the OR function.
This raises a natural question:
for an arbitrary adaptive $\ZZ_2$-linear measurement-based computation,
is there always such an AvN argument witnessing the contextuality being exploited as a computational resource?
At first sight, this appears challenging: if the protocol has nontrivial adaptivity, then the description of its output (the left-hand side of \cref{eq:ghz_condensed}) will not be a simple linear expression in measurements performed on the resource system.
For a minimal example, consider the composition of two Anders--Browne constructions by feeding the output of the first as input to the second to implement a three-wise OR.
The output is described by
\[x'_{x_a + y_b + z_{a+b}} + y'_c + z'_{x_a + y_b + z_{a+b}+c}\]
where the primed variables refer to measurements on the second copy of the GHZ model.
With greater adaptive depth, such expressions become increasingly complex.

In this paper, we resolve this question affirmatively.
The key insight is to consider AvN arguments over adaptive measurement protocols, which we treat as first-class measurements in their own right.
This allows us to construct an AvN argument explicitly for any adaptive $\ZZ_2$-linear MBQC, building the inconsistent systems of equations inductively from the structure of the computation tree.

\subsection*{Roadmap}

The remainder of the paper is organised as follows.
\Cref{sec:background} recalls the background on contextuality, in particular on $\ZZ_2$-linear AvN arguments.
\Cref{sec:mbqc} introduces $\ZZ_2$-linear measurement protocols on a Bell scenario, makes precise what it means for an MBQC to compute a Boolean function deterministically, and states our main result (\Cref{thm:main}) together with a proof sketch. \Cref{sec:illustrative} works through three examples of increasing complexity, recovering the PR-box and GHZ paradoxes and then treating a genuinely adaptive depth-three protocol in which each step of the general argument has a concrete counterpart.

The proof itself is set up in \Cref{sec:formalising}: it calls for a richer class of measurement protocols than the basic linear ones -- protocols that can halt on undesired outcomes or measure several sites simultaneously -- which are assembled into the scenario $\MP(\bell,\ZZ_2)$; the proof strategy is then presented, reducing the main theorem to a single identity (\Cref{thm:sumstozero}) whose proof proceeds by induction (roughly on protocol depth) requiring three key lemmas.
\Cref{sec:toolsfortheproof} proves the three lemmas underpinning this inductive argument, establishing in each case that protocol trees can be rearranged, condensed, or decomposed while preserving the sums at the heart of the argument.

\Cref{sec:cohomology} draws further consequences: most notably, it resolves Raussendorf's open question by extending cohomological witnesses of contextuality to the adaptive MBQC setting, in both the sheaf-cohomological (\Cref{thm:adaptiveAvNsandcohomology}) and the group-cohomological frameworks, with the \v{C}ech witness propagating all the way back to the original scenario. \Cref{sec:furtherdirections} concludes with directions for further work.

\section{Background on contextuality}
\label{sec:background}

We give a rapid pr\'ecis of the basic notions from \cite{ab} and \cite{abramskyetal:contextualitycohomologyparadox}.

\subsection{Scenarios, empirical models, and contextuality}\label{ssec:contextuality}
A contextuality \emph{scenario} is a structure $S\defeq (X, \ctx, O)$, where $X$ is a finite set of measurement labels or \emph{variables}; $\ctx$ is a family of subsets of $X$ giving the sets of \emph{compatible measurements} or \emph{contexts}; and $O$ is a set of measurement \emph{outcomes}.
The family $\ctx$ is required to be downward closed (subsets of compatible measurements are compatible) and to contain all singletons (any single measurement can be performed), \ie $\ctx$ is an abstract simplicial complex on $X$.\footnote{More generally, one may make $O$ into a dependent type, $O = \{ O_x \}_{x \in X}$, allowing the set of outcomes to depend  on the measurement. This is the setting of \cite{karvonen2018categories,barbosa2015dphil,abramskyetal:comonadicview}, though it will not be needed here.}

An important class of scenarios are the \emph{Bell scenarios}, which arise in non-locality arguments.
These are scenarios for which the set of measurements is partitioned as $X = \bigsqcup_{i \in \II} X_i$ and $C \in \ctx$ iff $|C \cap X_i | \leq 1$ for all $i \in \II$.
Thus $\II$ is a set of sites or agents (Alice, Bob, etc.), $X_i$ are the measurements available at site $i$, and the contexts arise by choosing at most one measurement from each site.

The \emph{event sheaf} over a scenario associates to each context $C \in \ctx$ the set $\ES(C) = O^C$ of joint outcomes for $C$, with restriction maps $\rho^{C'}_{C} : \ES(C') \to \ES(C)$ when $C \subseteq C'$ given by the obvious projection.
This is a sheaf $\ES : P^{\op} \to \CSet$ defined on the poset category $P = (\Pow(X),{\subseteq})$.
We write $\ls |_{C}$ for $\rho^{C'}_{C}(\ls)$ when $\ls \in \ES(C')$.

An \emph{empirical model} on a contextuality scenario is a family $e = \{ e_C \}_{C \in \ctx}$, where $e_C \in \Dist \ES(C)$ is a distribution over the joint outcomes for context $C$. 
We write $e:S$ to express that $e$ is an empirical model on the scenario $S$.
Here $\Dist$ is the distributions functor
that assigns to each set $X$ the set of discrete (finitely supported) distributions on $X$. This yields a presheaf $\Dist \ES : P^{\op} \to \CSet$.

To spell this out more explicitly, let $R$ be a commutative semiring of ``weights''. An $R$-distribution on $X$ is a function $d : X \to R$ with finite support $\supp(d) \defeq \{ x \in X \mid d(x) \neq 0 \}$ which satisfies the normalisation condition $\sum_{x \in X} d(x) = 1$.
The distribution functor $\Dist_{R}$, parameterised on $R$, sends $X$ to the set of all $R$-distributions on $X$.
The restriction maps now correspond to marginalisation. Note that, if we take $R$ to be the non-negative reals, we get the standard discrete probability distributions, while if we take $R$ to be the booleans, we get ``possibility distributions'', \ie  non-empty subsets, where non-emptiness corresponds to the normalisation condition.
This allows a unified treatment of probabilistic and ``probability-free'' non-locality and contextuality arguments.

Empirical models are required to satisfy the compatibility property: for all $C, C' \in \ctx$, $e_{C} |_{C \cap C'} = e_{C'} |_{C \cap C'}$. This says that the distributions over contexts have consistent marginals on their overlaps. 
This is known as the ``no-disturbance'' condition, and specialises on Bell scenarios to the standard ``no-signalling'' condition.

Whereas the event sheaf $\ES$ satisfies the sheaf property, so that compatible families over $\ES$ can always be glued together into a unique global section (an assignment of outcomes to all the measurements in the scenario), the presheaf $\Dist_{R} \ES$ is \emph{not} a sheaf.
Counter-examples to gluing correspond exactly to contextuality.\footnote{The fact that $\ES$ is a sheaf says that no-signalling deterministic models are non-contextual.}
More explicitly, an empirical model $e$ is \emph{non-contextual} iff there is a ``global distribution'' $d \in \Dist_{R} \ES(X)$ over global outcome assignments such that $d |_C = e_C$ for all $C \in \ctx$, \ie~iff a gluing exists for the compatible family.
Such a global distribution can be seen as a canonical form of ``hidden-variable model'' or ``ontological model''; see \cite{ab}.

In the case that $R$ is the non-negative reals, this yields the notion of \emph{probabilistic contextuality}, covering the usual examples from non-locality, typified by Bell's theorem \cite{bell1964einstein}. The case where $R$ is the booleans was dubbed \emph{logical contextuality} in \cite{ab}. This covers examples of ``probability-free'' non-locality proofs such as the Hardy paradox \cite{hardy1993nonlocality}, as observed in \cite{ab}. Note that logical contextuality is strictly stronger than probabilistic contextuality.
Indeed, a semiring homomorphism $R \to S$ induces a natural transformation $\Dist_R \Rightarrow \Dist_S$, which in turn lifts to a mapping from empirical models of $R$-distributions to empirical models of $S$-distributions.
Evidently, there is such a homomorphism from the non-negative reals to the booleans, which sends positive elements to $1$.\footnote{More generally, there is such a homomorphism from any zero-sum free semiring. Note, however, that there is no such homomorphism from the reals, or indeed any ring.}

There is a yet stronger, and in a suitable sense maximally strong, form of contextuality, known as strong contextuality. 
An empirical model $e$ is \emph{strongly contextual} iff there is no deterministic global assignment $\gs \in \ES(X)$ such that $\gs |_C \in \supp(e_C)$ for all $C \in \ctx$. To compare this with logical contextuality, note that the latter holds iff for \emph{some} local section $\ls \in \supp(e_C)$, there is no global assignment $\gs \in \ES(X)$ extending $\ls$
such that $\gs |_C \in \supp(e_C)$ for all $C \in \ctx$. By contrast, strong contextuality holds iff this is the case for \emph{all} local sections. 
Thus we obtain a strict hierarchy of notions of contextuality:
\[ \text{probabilistic} \; < \; \text{logical} \; < \; \text{strong} . \]

\subsection{AvN contextuality}\label{ssec:AvNcontextuality}
Now we assume that there is algebraic structure on the outcomes, specifically $O = \ZZ_2$, where we regard $\ZZ_2$ as an abelian group.\footnote{\label{fn:module}More generally, $O$ need only be a module over a ring for the linear theory to be defined; any abelian group is a module over $\ZZ$; cf.~\cite{abramskyetal:contextualitycohomologyparadox}.}
This algebraic structure can be used to define the \emph{$\ZZ_2$-linear theory} of an empirical model $e$.
Given a context $C$, a $\ZZ_2$-linear equation over $C$ has the form $(D, a)$, where $D = \{ x_1, \ldots , x_d \} \subseteq C$, and $a \in \ZZ_2$.
A section $\ls \in \ES(C)$ satisfies such an equation if $\sum_{x \in D} \ls(x) = a$.
We write this as $\ls \vDash \sum_{x \in D} x = a$, and for a model $e$, we write $e \vDash \sum_{x \in D} x = a$ iff $\ls \vDash \sum_{x \in D} x = a$ for all $\ls$ in the support of $e_C$.
The $\ZZ_2$-linear theory of an empirical model $e$ is the union over all contexts $C$ of the set of $\ZZ_2$-linear equations with variables in $C$ that are satisfied by all the sections in the support of $e_C$.  

This leads in turn to the notion of \emph{All-versus-Nothing contextuality}  \cite{abramskyetal:contextualitycohomologyparadox}. (The term ``All versus Nothing'' was first used in  \cite{mermin1990extreme}.)
An empirical model with outcomes $\ZZ_2$ is $\AvN_{\ZZ_2}$-contextual if its linear theory is inconsistent, \ie there is no global assignment of values in $\ZZ_2$ to all the variables which satisfies all the equations in the linear theory. Writing $\vDash^*$ for the closure of $\vDash$ under classical reasoning, such a model $e$ is  $\AvN_{\ZZ_2}$ iff $e\vDash^* 0=1$.
This is a more structured version of the theme of ``local consistency, global inconsistency'' \cite{abramsky2017contextuality}.
AvN contextuality is strictly stronger than strong contextuality \cite{abramskyetal:contextualitycohomologyparadox}.

\section{\texorpdfstring{$\mathbb{Z}_2$}{ℤ₂}-linear MBQC and the main result}\label{sec:mbqc}

In this section, we introduce the model of computation studied in this paper: measurement-based computation with $\ZZ_2$-linear classical control.
We define $\ZZ_2$-linear measurement protocols over a Bell scenario $\bell$ (\Cref{ssec:linearprotocols}) and how such a protocol, acting on an empirical model, computes a Boolean function (\Cref{ssec:computingwithlinearprotocols}).
We then state our main result -- that deterministically computing a non-affine function requires AvN contextuality -- and sketch an informal proof outline (\Cref{ssec:mainresult-pfsketch}).

\subsection{Linear protocols}\label{ssec:linearprotocols}

For a finite set of sites $\II$, let $\bell$ be the Bell scenario whose measurements are given by $X=\bigsqcup_{i \in \II} \ZZ_2$ with outcomes in $\ZZ_2$.  When useful, we denote the two measurements available at site $i$ by $x_{i,0}$ and $x_{i,1}$, although we often represent arbitrary measurements simply by $x,y,z$, or by $x_1,x_2,\ldots$ when referring to several at once.
We write $s\colon X \to \II$ for the map sending a measurement to its site. 

Given $U\subseteq V$, we write $\pi_U \colon \ZZ_2^V\to \ZZ_2^U$ for the projection and $\iota_V \colon \ZZ_2^U\to\ZZ_2^V$ for the inclusion, leaving the domain implicit. When $U$ is a singleton $\{i\}$, we write $\pi_i$ instead of $\pi_{\{i\}}$.

Roughly speaking, $\ZZ_2$-linear MBQC protocols operate on the scenario $\bell$ for some set of sites $\II$, going through (some of) the sites in some fixed order, with later measurement choices depending linearly on past outcomes.
We represent such a protocol as a tree 
\[\begin{forest}
[     $x$
[$q$, edge label={node[midway,left]{$0$}}]		[$v.q$,  edge label={node[midway,right]{$1$}}]]
\end{forest}\]
where $x$ is the initial measurement, $q$ is the continuation protocol, and $v$ is a vector specifying how to update the measurement settings in the continuation upon obtaining outcome $1$ for the measurement $x$.
The operational interpretation of such a tree is then a protocol that starts by measuring $x$, and then proceeds to either $q$ or $v.q$ depending on the observed outcome. 

This suggests an inductive definition, where one represents such a tree by the triple $\branching{x}{v}{q}$.
The following definition, our version of~\cite[Definition~1]{raussendorf:contextualityinmbqc}, makes this precise.

\begin{definition}\label{def:linearmp}
For a set of sites $U\subseteq\II$, we define
 the set $\MP_U$ of \emph{linear measurement protocols} that measure the sites $U$. We overload notation by writing $s$ for the function sending a protocol $p$ to the set $s(p)\subseteq \II$ of sites it measures.
These notions are defined inductively as follows:
\begin{itemize}
	\item We set $\MP_\emptyset\defeq\{()\}$, \ie there is a unique measurement protocol measuring no sites, namely the empty protocol $()$.
	\item If  $q\in \MP_{U}$ is a linear measurement protocol, $v$ is a vector in $\ZZ_2^{U}$ and $x$ is a measurement at a site not measured by $q$, \ie  $s(x)\notin U$,
	then $p\defeq\branching{x}{v}{q}$ is a linear measurement protocol and $s(p)=\{s(x)\}\cup U$, so that  $p\in \MP_{\{s(x)\} \cup U}$.
\end{itemize}
\end{definition}

The above definition encodes the data needed to describe a linear measurement protocol.
However, we also need to specify how suitably-sized vectors act on measurement protocols, so that the continuation $v.q$ is well defined. 

\begin{definition}\label{def:actiononlinearMPs}
The regular action of $\ZZ_2$ on itself induces a group action of $\ZZ_2$ on the set of measurements $\bigsqcup_{i \in \II} \ZZ_2$ of $\bell$: explicitly, this action is given by  $v.x_{i,a}=x_{i,a+v}$ for $v\in \ZZ_2$. 
Note that this action preserves the site the measurement lives on. 
We define an action of $\ZZ_2^U$ on $\MP_U$ by induction on the structure of the protocol. The action of $\ZZ_2^\emptyset\cong\{0\}$ is uniquely defined on $\MP_\emptyset$.  The action of $w\in \ZZ_2^{s(p)}$ on $p=\branching{x}{v}{q}$ is defined as $w.p\defeq\branching{\pi_{s(x)}(w).x}{v}{\pi_{s(q)}(w).q}$.
\end{definition}

We call these linear measurement protocols since the adaptivity in them is $\ZZ_2$-linear. While we later generalise this definition somewhat (allowing for some amount of non-linearity), we do not define measurement protocols in the full generality used in~\cite{abramskyetal:comonadicview}. 

The set $s(p)$ is linearly ordered  by the protocol $p$, with the order defined inductively for $\branching{x}{v}{q}$ by setting $s(x)<s(q)$, \ie $s(x)$ is below every element of $s(q)$, followed by the order on $s(q)$ induced by $q$. We will use this \emph{intrinsic order} on $s(p)$ throughout. 

We often represent linear measurement protocols as trees as in the beginning of this section. We use such visual notations more expressively in the sequel. For example, we write 
\[ \vcenter{\hbox{\begin{forest} [$p$ [$x$ [$q$, edge label={node[midway,left]{$0$}} ] [$v.q$, edge label={node[midway,right]{$1$}} ]]]\end{forest}}}\] 
to mean a protocol that performs the protocol $p$, then proceeds to $x$ and then to $q$ or $v.q$ depending on the outcome of $x$. However, this notation suppresses the fact that the outcome of $p$ may influence the measurement settings of $x$ and $q$, isolating just the action of $x$ on the sequel while reminding the reader that there indeed was a past prior to $x$.
\begin{definition}
Given a linear measurement protocol $p$, we define its \emph{adaptivity matrix} $T_p$ (or simply $T$), an $s(p)\times s(p)$ matrix over $\mathbb{Z}_2$, by induction on the structure of $p$. If $p$ is the empty protocol, then $T_p$ is the unique $0\times 0$ matrix. If $p=\branching{x}{v}{q}$, we define 
	\[T_p\defeq\begin{bNiceMatrix}[first-row,first-col]   &s(x) & s(q) \\
													 s(x)	&	0 & 0  \\
													s(q)&	v & T_q \end{bNiceMatrix}\]
\end{definition} 

Note that when $|s(p)|=1$, the definition above gives $T_p=\begin{bmatrix} 0 \end{bmatrix}$.
The operational meaning of the matrix is that the $i$th column specifies how the measurement outcome at site $i$ affects other measurement settings.
By construction, the matrix $T_p$ is strictly lower triangular when one orders $s(p)$ in the intrinsic order.
This constraint ensures that a measurement outcome can only affect the settings of (strictly) later measurements.

\begin{remark}\label{rem:reconstructinglinearMPfromamatrix}
The matrix $T_p$ does not fix the protocol $p$ uniquely.
First, the same matrix might be lower-triangular with respect to many different orderings. Any such ordering is the intrinsic order of some linear protocol with the same matrix.
However, even if we fix the linear order, we are still missing information.
This is because the indexing set only keeps track of the measurement site but not of the measurement setting.
This missing information can be provided by a vector $w_p\in \ZZ_2^{s(p)}$ specifying the default measurement settings (\ie those used when all prior outcomes are $0$). 
Then one can represent $p$ uniquely by $s(p)$ equipped with the intrinsic order of $p$, $T_p$, and $w_p$.
The action $(v,p)\mapsto v.p$ translates in this representation to $v.(T_p,w_p)=(T_p,v+w_p)$.
In~\cite{raussendorf:contextualityinmbqc}, the default measurement setting is (in a sense) fixed to be $0$, which corresponds to working with linear rather than affine functions of prior outcomes, so this issue does not arise. 
\end{remark}

\subsection{Computing with linear protocols}\label{ssec:computingwithlinearprotocols}

We discuss how such protocols, given access to an empirical model $e:\bell$, can be used to compute a Boolean function $f\colon \ZZ_2^l \to \ZZ_2^o$.
To do so, we first fix a base protocol $p$ and a linear input encoding $Q\colon \ZZ_2^l\to\ZZ_2^{s(p)}$: the idea is that given an input $w\in \ZZ_2^l$, we run the protocol $Q(w).p$ on our empirical model, obtaining a vector $\ZZ_2^{s(p)}$ of measurement outcomes (an outcome for each measurement performed while running $p$).
This outcome is then post-processed by another linear map $Z:\ZZ_2^{s(p)}\to \ZZ_2^o$ to produce the overall output of the computation.
This whole computation is then captured by fixing the empirical model $e:\bell$, the base protocol $p\in \MP$, and the pre- and post-processing maps $Q$ and $Z$.
We say that the MBQC $(e,p,Q,Z)$ computes $f$ deterministically if, regardless of the (in general probabilistic) outcomes obtained from $e$, every input $w$ in $\ZZ_2^l$ is always mapped to the correct output $f(w)\in\ZZ_2^o$.

This can be understood more formally as follows: once the protocol $p$ has been fixed, the empirical model induces a stochastic map $e_p\colon \ZZ_2^{s(p)}\to \ZZ_2^{s(p)}$ by sending a vector $v\in\ZZ_2^{s(p)}$ to the joint outcome in $\ZZ_2^{s(p)}$ obtained when running $v.p$.
Then, $(e,p,Q,Z)$ computing $f$ deterministically boils down to commutativity of
  \[
  \begin{tikzpicture}
     \matrix (m) [matrix of math nodes,row sep=1em,column sep=4em,minimum width=2em]
     {
      \ZZ_2^l & \ZZ_2^o \\
      \ZZ_2^{s(p)} & \ZZ_2^{s(p)} \\};
     \path[->]
     (m-1-1) edge node [left] {$Q$} (m-2-1)
            edge node [above] {$f$} (m-1-2)
     (m-2-2) edge node [right] {$Z$} (m-1-2)
     (m-2-1) edge node [below] {$e_p$} (m-2-2);
  \end{tikzpicture}
  \]
 
Now, roughly speaking, \cite[Theorem~2]{raussendorf:contextualityinmbqc} states that if an MBQC $(e,p,Q,Z)$ computes a non-affine function $f$, then $e$ is strongly contextual.
\footnote{More precisely, Raussendorf shows that $e$ is strongly contextual even when restricted to those joint measurements that could occur when running $p$.}
We strengthen this to show that $e$ is AvN-contextual when extended to a suitable scenario of measurement protocols.
We first make some further observations and reductions. 

As a function $f\colon \ZZ_2^l\to\ZZ_2^o$ is affine iff each of its coordinate functions is, it suffices to study functions $f\colon \ZZ_2^l\to\ZZ_2$ returning a single output bit.
We now show that we can further reduce to two input bits.
We say that a function $f\colon \ZZ^l_2\to \ZZ_2$
has \emph{even (odd) weight} if the preimage of $1$ under $f$ has even (odd) cardinality, \ie $\sum_{v\in \ZZ^l_2} f(v)=0$ ($=1$). While every affine function $\ZZ^l_2\to \ZZ_2$ has even weight when $l\geq 2$, the converse does not hold in general. However, when $l=2$, all functions with even weight are in fact affine. 
Moreover, since affineness is a condition stated with two variables, it can be tested on each two-dimensional subspace, so that a function is affine iff its restrictions to two-dimensional subspaces are affine. We thus obtain the following.

\begin{lemma}\label{lem:reductiontooddweight}
A function $f\colon \ZZ^l_2\to \ZZ_2$ is affine if and only if for every two-dimensional subspace $V\subseteq \ZZ^l_2$ the restriction $f|_V$ has even weight.
\end{lemma} 

Consider an MBQC $(e,p,Q,Z)$ that deterministically computes a non-affine function $f\colon \ZZ_2^l\to\ZZ_2$. As $f$ is not affine, by \Cref{lem:reductiontooddweight} some restriction to a two-dimensional subspace has odd weight. We can then precompose both $f$ and $Q$ with the inclusion $\ZZ_2^2\hookrightarrow \ZZ^l_2$ of this subspace, resulting in a function $\ZZ_2^2\to \ZZ_2$ with odd weight computed deterministically by the resulting measurement protocol. Thus we have reduced to the case when $l=2$.

\subsection{Main result and proof sketch}\label{ssec:mainresult-pfsketch}

To get an AvN argument in the sense of~\cite{abramskyetal:contextualitycohomologyparadox},
we need a measurement scenario where all measurements have outcomes valued in a single ring (or module).
The scenario of all linear measurement protocols on $\bell$ will not do,
since different measurement protocols can return bit-strings of different lengths.
This can be fixed by incorporating the post-processing $Z$ into the measurements of the scenario.
However, in order to obtain an AvN argument, we also need to generalise our notion of a measurement protocol.
A crucial modification is to include certain subprotocols of linear measurement protocols as measurements in the scenario: this allows us to decompose protocols into sums of smaller protocols that may in turn satisfy further linear equations.
The resulting scenario is defined later in \Cref{def:MPScenario}
and denoted $\MP(\bell,\ZZ_2)$: any model $e:\bell$ extends to a model $\MP(e,\ZZ_2):\MP(\bell,\ZZ_2)$.

Let us state our main result:
\begin{theorem}\label{thm:main}
     If an MBQC $(e,p,Q,Z)$ deterministically computes a non-affine function $f\colon \ZZ_2^l\to\ZZ_2^o$ then the model $\MP(e,\ZZ_2)$ on $\MP(\bell,\ZZ_2)$ induced by $e$ has an inconsistent $\ZZ_2$-linear theory.
\end{theorem}
The basic premise of the proof is simple. We first reduce to a function $\ZZ_2^{l}\to \ZZ_2$ by picking a non-affine coordinate function of $f$, and then restrict to a function $\hat{f}\colon \ZZ_2^2\to\ZZ_2$ with odd weight  by \Cref{lem:reductiontooddweight}.

Then, an MBQC $(e,p,Q,Z)$ deterministically computes a function $\hat{f}\colon \ZZ_2^2\to\ZZ_2$ iff
\begin{equation}\label{eq:computingf}\MP(e,\ZZ_2)\vDash Q(w).p;Z=\hat{f}(w) \;\text{ for all }w\in \ZZ_2^2,\end{equation}
where we (temporarily) write $Q(w).p;Z$ for the measurement in $\MP(\bell,\ZZ_2)$ which post-processes $Q(w).p$ by $Z$.
When $\hat{f}$ has odd weight, adding these four equations yields $1$ on the right-hand side, so that
\[\MP(e,\ZZ_2)\vDash^*\sum_{w\in \ZZ_2^2} Q(w).p;Z=1.\]
Hence, to prove \Cref{thm:main}, it suffices to show that 
\begin{equation}\label{eq:sumstozero}\MP(e,\ZZ_2)\vDash^*\sum_{w\in \ZZ_2^2} Q(w).p;Z=0\end{equation}
and that's where most of the effort lies.  In fact, this equation holds in the $\ZZ_2$-linear theory of $\MP(e,\ZZ_2)$ for \emph{any} $e:\bell$
(see \Cref{prop:allMPequations}). 
As a result, we will slightly abuse notation and often omit $\MP(e,\ZZ_2)\vDash$ and $\MP(e,\ZZ_2)\vDash^*$ from the equations we state.

We close this section by sketching the overall  structure of the inductive argument for \cref{eq:sumstozero}.
However, the reader should be aware that the sketch below is morally but not formally correct: roughly speaking, making it fully rigorous requires proving a more general result with a sufficiently strong induction hypothesis making everything work.

\begin{proof}[Proof sketch] 
The proof proceeds by induction on $n\defeq\left\vert s(p)\right\vert$, the number of sites in $p$.

If $n=1$, then $Q \colon \ZZ_2^2 \to \ZZ_2^1$ is not injective, since its domain has strictly larger dimension than its codomain.
As a result, each summand in \cref{eq:sumstozero} appears an even number of times and hence cancels out, so \cref{eq:sumstozero} holds trivially.

For $n>1$, the result follows from the following three lemmas: 
\begin{enumerate}
	\item If $\rank T_p=n$, we show in \Cref{lem:cansplice} that there is another measurement protocol $q$  such that 
		\[\sum_{w\in \ZZ_2^2} Q(w).p;Z \;=\, \sum_{w\in \ZZ_2^2} Q(w).q;Z\]
	and $\rank T_q<n$. In other words, without loss of generality, we can assume that $T_p$ is not of maximal rank.
	\item If $\rank T_p<n$, we show in \Cref{lem:nonmaximalrankacrossthesum} that we can condense $p$ to have strictly fewer levels without affecting the value of the sum $\sum_{w\in \ZZ_2^2} Q(w).p;Z$.
    The cost of this is that we need to allow \emph{higher-arity nodes}, which, rather than measuring a single measurement, measure the total parity $\sum x_i$ of a set of measurements in one step.
	\item Given a measurement protocol, we show in \Cref{lem:sumofsmallersums} that we can systematically replace all higher-arity branching nodes  by sums of protocols with unary nodes. This will let us write the overall sum as a sum of smaller sums $\sum_{w} Q_i(w).p_i;Z_i$, where each $p_i$ has strictly fewer than $n$  sites, therefore reducing to the induction hypothesis.
    However, the cost to pay is that we need to allow measurement protocols that have \emph{multiplicative nodes} which, instead of branching into two continuations, continue only if a certain fixed outcome is observed, otherwise returning $0$. While such nodes allow one to implement non-affine functions, it turns out that each level with such nodes increases the input dimension by $1$ as well, and the result remains true provided the input dimension is sufficiently large compared to the number of such levels. \qedhere
\end{enumerate}
\end{proof}

\section{Illustrative examples}\label{sec:illustrative}

In this section, we give examples of AvN arguments that arise in this manner.
We are not as precise here as in the formal treatment of the following sections, since the goal is to illustrate rather than to prove carefully. 

Let $\left\vert \II\right\vert=2$, and denote by $x_0,x_1$ the measurements at the first site and by $y_0,y_1$ those at the second site.
 Let the base protocol $p$ just measure $x_0$ and then $y_0$ regardless of the outcome, with the post-processing $Z$ returning the sum of the two outcome bits,
 and set the pre-processing $Q$ to be identity on $\ZZ_2^2$, \ie the first input bit determines the measurement choice for $x$ and the second input bit the measurement choice for $y$.
 If $e: \bell$ is an empirical model such that $(e,p,Q,Z)$ deterministically computes the AND function, equation~\eqref{eq:computingf} specialises to the following system of equations:
\[
\begin{aligned}
    \fitX{x_0} + \fitX{y_0}  &= 0 \qquad & 
    \fitX{x_0} + \fitX{y_1}  &= 0 \\
    \fitX{x_1} + \fitX{y_0}  &= 0 \qquad&
    \fitX{x_1} + \fitX{y_1}  &= 1
\end{aligned}
\]
The only empirical model satisfying these equations is the Popescu--Rohrlich (PR) box \cite{popescu1994quantum},
and thus \Cref{thm:main} recovers the usual AvN argument of the PR box (obtained by summing the four equations) as a special case.

Now, let $\left\vert \II\right\vert=3$, denoting the measurements at the first, second, and third sites by $x_a$, $y_a$, and $z_a$ for $a\in\ZZ_2$, respectively.
Let the base protocol $p$ measure $x_0$, then $y_0$, and finally $z_0$ (once again, adaptivity is trivial), with the  post-processing $Z$ returning the sum of the three outcome bits,
and let $Q$ be the  linear preprocessing function
$Q \colon \ZZ_2^2 \to \ZZ_2^3 \coloncolon  (a,b) \mapsto (a, b, a + b)$ from \Cref{sec:introduction}. If $e: \bell$ is an empirical model such that $(e,p,Q,Z)$ deterministically computes the OR function, equation~\eqref{eq:computingf} specialises to 
\[
\begin{aligned}
    \fitX{x_0} + \fitX{y_0} + \fitX{z_0} &= 0 \qquad & 
    \fitX{x_0} + \fitX{y_1} + \fitX{z_1} &= 1 \\
    \fitX{x_1} + \fitX{y_0} + \fitX{z_1} &= 1 \qquad&
    \fitX{x_1} + \fitX{y_1} + \fitX{z_0} &= 1
\end{aligned}
\]
\ie to \eqref{eq:ghz} from \Cref{sec:introduction}.
Thus \Cref{thm:main} recovers the usual AvN argument of the GHZ state as a special case.

We now turn to an example in which adaptivity plays a genuine role, again with $\left\vert \II\right\vert=3$.
The argument establishing~\eqref{eq:sumstozero} is correspondingly more involved: we sketch it here, indicating how each step of the argument corresponds to a general result formalised in the following sections.
We adopt the same notation as in the GHZ example, and let the base protocol $p$ start from $x_0$, then proceed to $y$, and finally to $z$, with the measurement settings for $y$ and $z$ given by the immediately preceding outcome.
The tree representing $p$ and the matrix $T_p$ (which has rank $2$) are given by 
\[\vcenter{\hbox{\begin{forest}
[     $x_0$
[$y_0$, edge label={node[midway,left]{$0$}}  
[$z_0$, edge label={node[midway,left]{$0$}}]		[$z_1$,  edge label={node[midway,right]{$1$}}]]		      [$y_1$,  edge label={node[midway,right]{$1$}}  
[$z_0$, edge label={node[midway,left]{$0$}}]		[$z_1$,  edge label={node[midway,right]{$1$}}] ]]
\end{forest}}}
\qquad\text{and}\qquad
T_p=\begin{bNiceMatrix}[first-row,first-col]    & x & y & z \\
													 x	& 0 & 0 & 0 \\
													 y & 1 & 0 & 0 \\
                                                  z & 0 & 1 & 0  \end{bNiceMatrix}.\]
We set the outcome post-processing to return only the final bit, \ie the outcome of the $z_a$ measurement, and define the input encoding $Q\colon \ZZ^2_2\to \ZZ_2^3$ by  $(a,b) \mapsto (a, b, b)$.
Now, if we assume that we are computing the AND function using these choices for $p,Q,Z$, we obtain the following equations:
\begin{subequations}
\label{eqs:depth3trees}
\begin{center}
\renewcommand{\theequation}{\arabic{parentequation}--\roman{equation}}
\begin{tabular}{cc}
\begin{minipage}{0.45\linewidth}
\begin{equation}\label{eq:tree-i}
\vcenter{\hbox{
\begin{forest}
[$x_0$
    [$y_0$, edge label={node[midway,left]{$0$}}
        [$z_0$, edge label={node[midway,left]{$0$}}]
        [$z_1$, edge label={node[midway,right]{$1$}}]
    ]
    [$y_1$, edge label={node[midway,right]{$1$}}
        [$z_0$, edge label={node[midway,left]{$0$}}]
        [$z_1$, edge label={node[midway,right]{$1$}}]
    ]
]
\end{forest}
}}
= 0
\end{equation}
\end{minipage}
&
\begin{minipage}{0.45\linewidth}
\begin{equation}\label{eq:tree-ii}
\vcenter{\hbox{
\begin{forest}
[$x_1$
    [$y_0$, edge label={node[midway,left]{$0$}}
        [$z_0$, edge label={node[midway,left]{$0$}}]
        [$z_1$, edge label={node[midway,right]{$1$}}]
    ]
    [$y_1$, edge label={node[midway,right]{$1$}}
        [$z_0$, edge label={node[midway,left]{$0$}}]
        [$z_1$, edge label={node[midway,right]{$1$}}]
    ]
]
\end{forest}
}}
= 0
\end{equation}
\end{minipage}
\\[3ex]
\begin{minipage}{0.45\linewidth}
\begin{equation}\label{eq:tree-iii}
\vcenter{\hbox{
\begin{forest}
[$x_0$
    [$y_1$, edge label={node[midway,left]{$0$}}
        [$z_1$, edge label={node[midway,left]{$0$}}]
        [$z_0$, edge label={node[midway,right]{$1$}}]
    ]
    [$y_0$, edge label={node[midway,right]{$1$}}
        [$z_1$, edge label={node[midway,left]{$0$}}]
        [$z_0$, edge label={node[midway,right]{$1$}}]
    ]
]
\end{forest}
}}
= 0
\end{equation}
\end{minipage}
&
\begin{minipage}{0.45\linewidth}
\begin{equation}\label{eq:tree-iv}
\vcenter{\hbox{
\begin{forest}
[$x_1$
    [$y_1$, edge label={node[midway,left]{$0$}}
        [$z_1$, edge label={node[midway,left]{$0$}}]
        [$z_0$, edge label={node[midway,right]{$1$}}]
    ]
    [$y_0$, edge label={node[midway,right]{$1$}}
        [$z_1$, edge label={node[midway,left]{$0$}}]
        [$z_0$, edge label={node[midway,right]{$1$}}]
    ]
]
\end{forest}
}}
= 1
\end{equation}
\end{minipage}
\\
\end{tabular}
\end{center}
\end{subequations}
These equations are slightly informal and in particular, they omit the additional detail that the protocols return only a single bit (the outcome of the final $z_a$ measurement), so we must remember this ourselves.
Again, the right-hand sides sum to $1$; we thus wish to show that the left-hand sides sum to $0$, establishing that \cref{eq:sumstozero} holds.
But to do so, we need to evaluate such sums of trees.

The first key identity is that a tree equals the sum of its branches; for the first tree in~\eqref{eqs:depth3trees}, this reads:
\[\vcenter{\hbox{\begin{forest}
[$x_0$[$y_0$, edge label={node[midway,left]{$0$}} [$z_0$, edge label={node[midway,left]{$0$}}][$z_1$,  edge label={node[midway,right]{$1$}}]][$y_1$,  edge label={node[midway,right]{$1$}} [$z_0$, edge label={node[midway,left]{$0$}}]	[$z_1$,  edge label={node[midway,right]{$1$}}] ]]\end{forest}}}
=\vcenter{\hbox{\begin{forest} [$x_0$ [$y_0$, edge label={node[midway,right]{$0$}} [$z_0$, edge label={node[midway,right]{$0$}} ]]]\end{forest}}} + 
	\vcenter{\hbox{\begin{forest} [$x_0$ [$y_0$, edge label={node[midway,right]{$0$}} [$z_1$, edge label={node[midway,right]{$1$}} ]]]\end{forest}}} +
    \vcenter{\hbox{\begin{forest} [$x_0$ [$y_1$, edge label={node[midway,right]{$1$}} [$z_0$, edge label={node[midway,right]{$0$}} ]]]\end{forest}}}
    + 
	\vcenter{\hbox{\begin{forest} [$x_0$ [$y_1$, edge label={node[midway,right]{$1$}} [$z_1$, edge label={node[midway,right]{$1$}} ]]]\end{forest}}}\]
To explain the meaning of such an equation, let us first look at the interpretation of a single branch, say the leftmost one:
\[\begin{forest} [$x_0$ [$y_0$, edge label={node[midway,right]{$0$}} [$z_0$, edge label={node[midway,right]{$0$}} ]]]\end{forest}\]

It depicts a (no longer linear) measurement protocol, which starts from the top and attempts to obtain the indicated outcomes in order: if it succeeds, it returns the outcome of the final measurement $z_0$; if at any point a measurement returns a non-indicated outcome (\eg measuring $x_0$ yields outcome $1$), the protocol halts and returns $0$.
Both the tree itself and its branches are measurements in their own right in $\MP(\bell,\ZZ_2)$.
Asserting an equation like this is to claim that the support of our empirical model (extended to $\MP(\bell,\ZZ_2)$) satisfies it, \ie that upon measuring both sides of the equation, the obtained joint outcome satisfies the equation. 

While all equations we discuss could be expressed by breaking every protocol into its branches, it is conceptually clearer to group branches into larger subprotocols.
Hence our base protocol also satisfies the equation
\[\vcenter{\hbox{\begin{forest}
[$x_0$[$y_0$, edge label={node[midway,left]{$0$}} [$z_0$, edge label={node[midway,left]{$0$}}][$z_1$,  edge label={node[midway,right]{$1$}}]][$y_1$,  edge label={node[midway,right]{$1$}} [$z_0$, edge label={node[midway,left]{$0$}}]	[$z_1$,  edge label={node[midway,right]{$1$}}] ]]\end{forest}}}
=\vcenter{\hbox{\begin{forest}[$x_0$[$y_0$, edge label={node[midway,left]{$0$}} [$z_0$, edge label={node[midway,left]{$0$}}][$z_1$,  edge label={node[midway,right]{$1$}}]]]\end{forest}}}
+ \vcenter{\hbox{\begin{forest}[$x_0$[$y_1$, edge label={node[midway,left]{$1$}} [$z_0$, edge label={node[midway,left]{$0$}}][$z_1$,  edge label={node[midway,right]{$1$}}]]]\end{forest}}}\]
asserting that the whole protocol equals the sum of its left and right halves.
A similar equation is satisfied by each of the four trees in~\eqref{eqs:depth3trees}.
As the aim is to show that the sum of these four trees is zero, when evaluating this sum we can regroup these halves as we wish.
In particular, if we add the left half of the base protocol~\eqref{eq:tree-i} and the right half of~\eqref{eq:tree-iii}, we obtain another linear measurement protocol:
\[\vcenter{\hbox{\begin{forest}[$x_0$[$y_0$, edge label={node[midway,left]{$0$}} [$z_0$, edge label={node[midway,left]{$0$}}][$z_1$,  edge label={node[midway,right]{$1$}}]]]\end{forest}}}
+ \vcenter{\hbox{\begin{forest}[$x_0$[$y_0$, edge label={node[midway,left]{$1$}} [$z_1$, edge label={node[midway,left]{$0$}}][$z_0$,  edge label={node[midway,right]{$1$}}]]]\end{forest}}}
=\vcenter{\hbox{\begin{forest}
[$x_0$[$y_0$, edge label={node[midway,left]{$0$}} [$z_0$, edge label={node[midway,left]{$0$}}][$z_1$,  edge label={node[midway,right]{$1$}}]][$y_0$,  edge label={node[midway,right]{$1$}} [$z_1$, edge label={node[midway,left]{$0$}}]	[$z_0$,  edge label={node[midway,right]{$1$}}] ]]\end{forest}}}
 \] 
Such re-groupings are formalised in \Cref{lem:splicing,lem:cansplice}.

This recombined measurement protocol has a different adaptivity structure from the original ones:
after measuring $x_0$, it measures $y_0$ regardless of the outcome, while the final measurement setting for $z$ is determined by the parity of the first two outcomes.
The corresponding adaptivity matrix is
\[\begin{bNiceMatrix}[first-row,first-col]    & x & y & z \\
													 x	& 0 & 0 & 0 \\
													 y & 0 & 0 & 0 \\
                                                  z & 1 & 1 & 0  \end{bNiceMatrix},\]
which now has rank $1$.
We could therefore condense this information and represent the measurement protocol as 
\[\begin{forest}
[     $x_0+y_0$
[$z_0$, edge label={node[midway,left]{$0$}}]		[$z_1$,  edge label={node[midway,right]{$1$}}]]
\end{forest}\]
which first measures both $x_0$ and $y_0$, then, based on their parity $x_0+y_0$,
selects the setting to use for $z$, and finally returns the outcome of this last measurement.
Such measurement protocols that may measure a sum of basic measurements in a single step are defined formally in \Cref{def:moregeneralmps}, and this operation of condensing a longer protocol into fewer levels is formalised in \Cref{lem:nonmaximalrank}.

Performing this regrouping for all four trees in~\eqref{eqs:depth3trees},
we find that showing equation~\eqref{eq:sumstozero} holds for them is equivalent to showing that it holds for the following four trees:
\begin{equation}\label{pic:simplertrees}
\vcenter{\hbox{\begin{forest}
[ $x_0+y_0$
[$z_0$, edge label={node[midway,left]{$0$}}]		[$z_1$,  edge label={node[midway,right]{$1$}}]] \end{forest}}}
 \qquad
\vcenter{\hbox{\begin{forest}[$x_1+y_0$ 
[$z_0$, edge label={node[midway,left]{$0$}}]		[$z_1$,  edge label={node[midway,right]{$1$}}]]
\end{forest}}}
 \qquad
\vcenter{\hbox{\begin{forest}[$x_0+y_1$ 
[$z_1$, edge label={node[midway,left]{$0$}}]		[$z_0$,  edge label={node[midway,right]{$1$}}]]
\end{forest}}}
 \qquad
\vcenter{\hbox{\begin{forest}[$x_1+y_1$ 
[$z_1$, edge label={node[midway,left]{$0$}}]		[$z_0$,  edge label={node[midway,right]{$1$}}]]
\end{forest}}}\end{equation}
In summary, we have regrouped the four original measurement protocols (without affecting their total sum) to obtain four protocols with simpler adaptivity, at the cost of allowing protocols that measure a sum of basic measurements in a single step.

Now, to show that these sum to zero, we first assert the equation 
            \[ \vcenter{\hbox{\begin{forest} [$x_a +y_b$ [$z_c$, edge label={node[midway,right]{$r$}} ]]\end{forest}}}
	= \vcenter{\hbox{\begin{forest} [$x_a$ [$z_c$, edge label={node[midway,right]{$r$}} ]]\end{forest}}} + 
	\vcenter{\hbox{\begin{forest} [$y_b$ [$z_c$, edge label={node[midway,right]{$1$}} ]]\end{forest}}}
	\]
(true in a measurement context of the scenario of measurement protocols on $\bell$; see proof of \Cref{lem:separatingsums})
for each branch.
Applying this to the first tree in \eqref{pic:simplertrees} yields
\[
\vcenter{\hbox{\begin{forest} [ $x_0+y_0$[$z_0$, edge label={node[midway,left]{$0$}}]		[$z_1$,  edge label={node[midway,right]{$1$}}]]\end{forest}}}=
  \vcenter{\hbox{\begin{forest} [$x_0 +y_0$ [$z_0$, edge label={node[midway,right]{$0$}} ]]\end{forest}}} + \vcenter{\hbox{\begin{forest} [$x_0 +y_0$ [$z_1$, edge label={node[midway,right]{$1$}} ]]\end{forest}}}
  = \vcenter{\hbox{\begin{forest} [$x_0$ [$z_0$, edge label={node[midway,right]{$0$}} ]]\end{forest}}} + 
	\vcenter{\hbox{\begin{forest} [$y_0$ [$z_0$, edge label={node[midway,right]{$1$}} ]]\end{forest}}}
	+ \vcenter{\hbox{\begin{forest} [$x_0$ [$z_1$, edge label={node[midway,right]{$1$}} ]]\end{forest}}} + 
	\vcenter{\hbox{\begin{forest} [$y_0$ [$z_1$, edge label={node[midway,right]{$1$}} ]]\end{forest}}}
=\vcenter{\hbox{\begin{forest} [$x_0$ [$z_0$, edge label={node[midway,left]{$0$}} ] [$z_1$, edge label={node[midway,right]{$1$}} ]]\end{forest}}}+
		\vcenter{\hbox{\begin{forest} [$y_0$ [$z_0$, edge label={node[midway,right]{$1$}} ]]\end{forest}}}+
	\vcenter{\hbox{\begin{forest} [$y_0$ [$z_1$, edge label={node[midway,right]{$1$}} ]]\end{forest}}}
\]
This kind of manipulation is demonstrated more carefully and in greater generality in \Cref{lem:separatingsums}.
Proceeding similarly for all four trees in~\eqref{pic:simplertrees} results in  
\begin{center}
\begin{tabular}{cc}
\begin{minipage}{0.45\textwidth}
\[
\vcenter{\hbox{\begin{forest}[ $x_0+y_0$[$z_0$, edge label={node[midway,left]{$0$}}]		[$z_1$,  edge label={node[midway,right]{$1$}}]]\end{forest}}}	
=\vcenter{\hbox{\begin{forest} [$x_0$ [$z_0$, edge label={node[midway,left]{$0$}} ] [$z_1$, edge label={node[midway,right]{$1$}} ]]\end{forest}}}+
		\vcenter{\hbox{\begin{forest} [$y_0$ [$z_0$, edge label={node[midway,right]{$1$}} ]]\end{forest}}}+
	\vcenter{\hbox{\begin{forest} [$y_0$ [$z_1$, edge label={node[midway,right]{$1$}} ]]\end{forest}}}
\]
\end{minipage}
&
\begin{minipage}{0.45\textwidth}
\[
\vcenter{\hbox{\begin{forest}[ $x_1+y_0$[$z_0$, edge label={node[midway,left]{$0$}}]		[$z_1$,  edge label={node[midway,right]{$1$}}]]\end{forest}}}	
=\vcenter{\hbox{\begin{forest} [$x_1$ [$z_0$, edge label={node[midway,left]{$0$}} ] [$z_1$, edge label={node[midway,right]{$1$}} ]]\end{forest}}}+
		\vcenter{\hbox{\begin{forest} [$y_0$ [$z_0$, edge label={node[midway,right]{$1$}} ]]\end{forest}}}+
	\vcenter{\hbox{\begin{forest} [$y_0$ [$z_1$, edge label={node[midway,right]{$1$}} ]]\end{forest}}}
\]
\end{minipage}
\\[3ex]

\begin{minipage}{0.45\textwidth}
\[
\vcenter{\hbox{\begin{forest}[ $x_0+y_1$[$z_1$, edge label={node[midway,left]{$0$}}]		[$z_0$,  edge label={node[midway,right]{$1$}}]]\end{forest}}}	
=\vcenter{\hbox{\begin{forest} [$x_0$ [$z_1$, edge label={node[midway,left]{$0$}} ] [$z_0$, edge label={node[midway,right]{$1$}} ]]\end{forest}}}+
		\vcenter{\hbox{\begin{forest} [$y_1$ [$z_1$, edge label={node[midway,right]{$1$}} ]]\end{forest}}}+
	\vcenter{\hbox{\begin{forest} [$y_1$ [$z_0$, edge label={node[midway,right]{$1$}} ]]\end{forest}}}
\]
\end{minipage}
&
\begin{minipage}{0.45\textwidth}
\[
\vcenter{\hbox{\begin{forest}[ $x_1+y_1$[$z_1$, edge label={node[midway,left]{$0$}}]		[$z_0$,  edge label={node[midway,right]{$1$}}]]\end{forest}}}	
=\vcenter{\hbox{\begin{forest} [$x_1$ [$z_1$, edge label={node[midway,left]{$0$}} ] [$z_0$, edge label={node[midway,right]{$1$}} ]]\end{forest}}}+
		\vcenter{\hbox{\begin{forest} [$y_1$ [$z_1$, edge label={node[midway,right]{$1$}} ]]\end{forest}}}+
	\vcenter{\hbox{\begin{forest} [$y_1$ [$z_0$, edge label={node[midway,right]{$1$}} ]]\end{forest}}}
\]
\end{minipage}
\\
\end{tabular}
\end{center}
Therefore, to evaluate the sum of the four trees in~\eqref{pic:simplertrees}, we may instead add the right-hand sides of the four equations above.
These split neatly into protocols that start from $x$ and others that start from $y$. 

The latter sum to zero as each branch occurs twice. As for the summands starting with $x$, they form a sum similar to the one we started with, but involving fewer variables. 
In the context of the general proof, this is then handled by the induction hypothesis.
Unwinding the resulting construction, we proceed by repeating the earlier trick of splitting trees into halves and rearranging. Specifically, when adding the two trees starting with $x_0$ we can compute as follows:
\begin{align*}
	&\vcenter{\hbox{\begin{forest} [$x_0$ [$z_0$, edge label={node[midway,left]{$0$}} ] [$z_1$, edge label={node[midway,right]{$1$}} ]] \end{forest}}}+ 
	\vcenter{\hbox{\begin{forest} [$x_0$ [$z_1$, edge label={node[midway,left]{$0$}} ] [$z_0$, edge label={node[midway,right]{$1$}} ]]\end{forest}}}
	=
	\vcenter{\hbox{\begin{forest} [$x_0$ [$z_0$, edge label={node[midway,right]{$0$}} ]]\end{forest}}} + 
	\vcenter{\hbox{\begin{forest} [$x_0$ [$z_1$, edge label={node[midway,right]{$1$}} ]]\end{forest}}} +
	\vcenter{\hbox{\begin{forest} [$x_0$ [$z_1$, edge label={node[midway,right]{$0$}} ]]\end{forest}}} + 
	\vcenter{\hbox{\begin{forest} [$x_0$ [$z_0$, edge label={node[midway,right]{$1$}} ]]\end{forest}}}
    \\
    &=
	\vcenter{\hbox{\begin{forest} [$x_0$ [$z_0$, edge label={node[midway,right]{$0$}} ]]\end{forest}}} + 
	\vcenter{\hbox{\begin{forest} [$x_0$ [$z_0$, edge label={node[midway,right]{$1$}} ]]\end{forest}}} +
	\vcenter{\hbox{\begin{forest} [$x_0$ [$z_1$, edge label={node[midway,right]{$0$}} ]]\end{forest}}} + 
	\vcenter{\hbox{\begin{forest} [$x_0$ [$z_1$, edge label={node[midway,right]{$1$}} ]]\end{forest}}}
	=
	\vcenter{\hbox{\begin{forest} [$x_0$ [$z_0$, edge label={node[midway,left]{$0$}} ] [$z_0$, edge label={node[midway,right]{$1$}} ]] \end{forest}}}+ 
	\vcenter{\hbox{\begin{forest} [$x_0$ [$z_1$, edge label={node[midway,left]{$0$}} ] [$z_1$, edge label={node[midway,right]{$1$}} ]]\end{forest}}}
	= z_0+z_1,
\end{align*}
where we used the equation $\vcenter{\hbox{\begin{forest} [$x_0$ [$z_c$, edge label={node[midway,left]{$0$}} ] [$z_c$, edge label={node[midway,right]{$1$}} ]] \end{forest}}}
	= z_c$.
This boils down to the fact that if a protocol returns the outcome of $z_c$ in every branch, then it might as well just measure $z_c$ directly. 
Similarly, the two trees starting with $x_1$ sum to $z_0+z_1$, cancelling the above.
Taken together, these computations show that the four trees in \eqref{eqs:depth3trees} sum to zero, concluding the AvN argument.
This kind of step is systematised in \Cref{lem:sumofsmallersums}.

The general proof proceeds along similar lines, \ie by simplifying the expressions and regrouping them until each grouping sums to zero because every expression in it occurs an even number of times. The main ways of simplifying expressions amount to:
\begin{itemize}
	\item splitting an expression into two halves, and recombining the halves with those from another tree (\Cref{lem:cansplice})
	\item simplifying the adaptivity at the cost of measuring sums of several measurements at once (\Cref{lem:nonmaximalrankacrossthesum})
	\item splitting a level involving a sum of variables into a sum of expressions (\Cref{lem:separatingsums}) and using this to write the whole sum as a sum of smaller things (\Cref{lem:sumofsmallersums}).
\end{itemize}
We make these steps rigorous in the next sections.

\section{Formalising the main result} 
\label{sec:formalising}

In this section, we formalise our main result.
We begin by extending the class of measurement protocols from \Cref{sec:mbqc}, then assemble these into a contextuality scenario, and finally state and prove the key equation~\eqref{eq:sumstozero} on which \Cref{thm:main} rests, modulo some lemmas deferred to the next section.
\subsection{Extended measurement protocols}
We change our measurement protocols in three ways:
\begin{itemize}
	\item We incorporate the single output bit returned by the protocol (\ie the function $Z\colon \ZZ_2^{s(p)}\to \ZZ_2$) into the protocol itself. 
	\item We allow for two kinds of nodes:
    \emph{branching nodes} (those considered so far), which split (linearly) into two continuations,
    and \emph{multiplicative nodes}, which continue only if a particular outcome is observed, otherwise returning $0$.   
	\item We allow nodes to have arbitrary arity, simultaneously performing a finite set $D=\{x_1,\ldots,x_d\}$ of measurements at distinct sites, rather than a single measurement at a single site. The operational meaning is that $D$ stands for performing all the measurements in $D$ and returning only their total parity $x_1 + \cdots + x_d \in \ZZ_2$ (rather than their individual outcomes).
\end{itemize}
The incorporation of the output bit and higher-arity nodes are technical conveniences that simplify the inductive argument; multiplicative nodes, on the other hand, genuinely extend the expressive power of protocols and play an essential role in the argument. 
Although these are less general than the measurement protocols of~\cite{abramskyetal:comonadicview}, we call them simply \emph{measurement protocols}, as the more general notion does not arise in this work.

From now on, given sets $U$ and $V$, we write $U + V$ for their disjoint union, and moreover we denote a singleton set by $1$.
In particular, given a set of sites $U$, $U+1$ stands for $U$ with an additional point $\ast$, which we use to index the output bit of a protocol.

\begin{definition}\label{def:moregeneralmps}
For sets of sites $U,V\subseteq \II$ (not necessarily disjoint) we define the set $\MP_{U,V}$ of measurement protocols
with multiplicative nodes at sites in $U$ and branching nodes at sites in $V$. 
For a protocol $p\in \MP_{U,V}$, we write $s(p) \defeq U\cup V$ for the set of sites it measures, and also define:
	 \begin{itemize}
	 		\item the set  $S_p\subseteq\Pow (s(p))$ of jointly measured sets of sites;
            \item the set $M_p\subseteq X$ of possibly performed measurements.
	\end{itemize}
    These notions are defined inductively. Call a set of measurements $D = \{x_1, \dots, x_d\}$ on distinct sites $s(D) \defeq \{s(x_1),\dots ,s(x_d)\}$ \emph{admissible for} $q \in \MP_{U,V}$ if
\begin{enumerate}
      \item $s(D)\notin S_q$, \ie a measurement on the exact same set of sites is not performed in $q$; 
      \item $\{x\in M_q \mid s(x)\in s(D)\}\subseteq D$, \ie if a site in $s(D)$ is possibly measured in $q$, it is always measured with the same setting as in $D$.
\end{enumerate}
The inductive clauses are then as follows:
\begin{itemize}
	\item We set $\MP_{\emptyset,\emptyset}\defeq\ZZ_2 = \{0,1\}$. For $p \in \MP_{\emptyset,\emptyset}$, we set $S_p, M_p\defeq\emptyset$.
 	\item If $q\in \MP_{U,V}$, $v\in \ZZ_2^{s(q)+1}$, and $D = \{x_1,\dots, x_d\}$ is admissible for $q$ and  $\pi_{s(D)\cap s(q)}v=0$,
	then $p\defeq\branching{D}{v}{q}$ is in $\MP_{U,V\cup s(D)}$. We call this a $d$-ary \emph{branching node}.
	We set $S_p\defeq S_q\cup\{s(D)\}$ and $M_p\defeq M_q\cup v.M_q\cup D$, where $v.M_q\defeq \{\pi_{s(x)} (v).x\mid x\in M_q\}$
 	\item If $q\in \MP_{U,V}$, $r\in \ZZ_2$, and $D = \{x_1,\dots, x_d\}$ is admissible for $q$,
	then $p\defeq\multiplicative{D}{r}{q}$ is a measurement protocol in $\MP_{U\cup s(D),V}$.
    We call this a $d$-ary \emph{multiplicative node}. We set $S_p\defeq S_q\cup\{s(D)\}$ and $M_p\defeq M_q\cup D$. 
\end{itemize}
We call a node (whether multiplicative or branching) \emph{unary} if $d=1$.
A \emph{branch} is a measurement protocol with only multiplicative nodes ending in $1$.
\end{definition}

We now extend the action from \Cref{def:actiononlinearMPs}.

\begin{definition}\label{def:actiononMPs}
We define an action of $\ZZ_2^{U\cup V+1}$ on $\MP_{U,V}$ by induction on the structure of measurement protocols. 
\begin{itemize}
\item For the base case, the action of $w \in \ZZ_2^1$ on $p \in \MP_{\emptyset,\emptyset} = \ZZ_2$ is given by $w.p \defeq w + p$, \ie  the regular action of $\ZZ_2$ on itself.
\item For a branching node, the action of $w\in \ZZ_2^{s(p)+1}$ on $p=\branching{D}{v}{q}$ is defined as
\[w.p \defeq \branching{\{\pi_{s(x)}(w).x \mid x \in D\}}{v}{\pi_{s(q)+1}(w).q}.\]
\item For a multiplicative node, the action of  $w\in \ZZ_2^{s(p)+1}$ on $p=\multiplicative{D}{r}{q}$ is defined as
\[w.p\defeq \multiplicative{\{\pi_{s(x)}(w).x \mid x \in D\}}{r}{\pi_{s(q)+1}(w).q}.\]
\end{itemize}
\end{definition}

Measurement protocols $\branching{D}{v}{q}$ (starting with a branching node) and $\multiplicative{D}{r}{q}$ (starting with a multiplicative node), with $D = \{x_1,\dots,x_d\}$, are respectively represented as
\[\vcenter{\hbox{\begin{forest}
[     $x_1+\dots + x_d$
[$q$, edge label={node[midway,left]{$0$}}]		[$v.q$,  edge label={node[midway,right]{$1$}}]]
\end{forest}}}
\qquad \text{and} \qquad
\vcenter{\hbox{\begin{forest} [$x_1+\dots + x_d$ [$q$, edge label={node[midway,right]{$r$}} ]]\end{forest}}}\]

In the first case, the operational interpretation is that the protocol starts by measuring $x_1,\dots, x_d$, then proceeds to either $q$ or $v.q$ depending on the total parity  $x_1+\dots +x_d$.
Note that the side-condition that $\pi_{s(D)\cap s(q)}v=0$ ensures that the outcome of $D$ cannot affect the measurement settings of sites in $s(D)$ (which have already been measured). Under the assumption that $D$ is admissible for $q$, this is equivalent to $D$ also being admissible for $v.q$.

In the second case, the interpretation is that $p$ measures $x_1,\dots, x_d$: if the total parity equals $r$, it proceeds to $q$; otherwise it returns $0$ immediately. 
For the cases $r=0$ and $r=1$, we could (but won't) depict a protocol $\multiplicative{D}{r}{q}$ starting with a multiplicative node  as 
\[\vcenter{\hbox{\begin{forest}[$x_1+\dots + x_d$ [$q$, edge label={node[midway,left]{$0$}}]	[$0$,  edge label={node[midway,right]{$1$}}]]\end{forest}}}\qquad \text{and} \qquad
\vcenter{\hbox{\begin{forest}[$x_1+\dots + x_d$ [$0$, edge label={node[midway,left]{$0$}}]	[$q$,  edge label={node[midway,right]{$1$}}]]\end{forest}}} \qquad \text{respectively.}\] 

\begin{definition}\label{def:numberofbranchingandmultlevels} 
For a protocol $p$, we define the \emph{number of multiplicative levels} $m(p)$ and the \emph{number of branching levels $n(p)$}
by induction on the structure of $p$:
\[
m(p) \defeq \begin{cases}
  0 & \text{if $p \in \MP_{\emptyset,\emptyset}$,} \\
  m(q) & \text{if $p = \branching{D}{v}{q}$,} \\
  m(q)+1 & \text{if $p = \multiplicative{D}{r}{q}$;}
\end{cases}
\qquad
n(p) \defeq \begin{cases}
  0 & \text{if $p \in \MP_{\emptyset,\emptyset}$,} \\
  n(q)+1 & \text{if $p = \branching{D}{v}{q}$,} \\
  n(q) & \text{if $p = \multiplicative{D}{r}{q}$.}
\end{cases}
\]
\end{definition}

We now slightly generalise the definition of the matrix associated to a linear measurement protocol, to accommodate unary multiplicative nodes as well.

\begin{definition}\label{def:adaptivitymatrix}
Given a measurement protocol $p\in\MP_{U,V}$ with all nodes unary,\footnote{Since all nodes are unary, $U$ and $V$ are necessarily disjoint, so $s(p)=U \cup V \cong U + V$. The restriction to unary nodes is not necessary to define the matrix, but it does make things technically simpler and is sufficiently general for our purposes.}
we define the \emph{adaptivity matrix} $T_p$ (or simply $T$) of $p$, a $(U+V+1)\times V$ matrix over $\mathbb{Z}_2$, by induction on the structure of $p$:
\begin{itemize}
 \item If $p \in \MP_{\emptyset,\emptyset}$, then $T_p$ is the unique  $1\times 0$ matrix.
 \item If $p = \branching{x}{v}{q}$ starts with a branching node where $q \in \MP_{U,V}$, we define
    \[T_p \defeq\; \begin{bNiceMatrix}[first-row,first-col]
            & s(x) & V\\ 
            s(x) & 0 & 0\\
            U+V+1 & v & T_q \end{bNiceMatrix}.
    \]
 \item If $p=\multiplicative{x}{r}{q}$ starts with a multiplicative node where $q \in \MP_{U,V}$, we define
    \[T_p \defeq\; \begin{bNiceMatrix}[first-row,first-col] 
            & V\\
            s(x) & 0\\
            U+V+1 & T_q \end{bNiceMatrix}.
    \]
\end{itemize}
The largest possible rank for the adaptivity matrix $T_p$ is $n(p)$, the number of branching levels; we say that $T_p$ (or $p$) is \emph{of maximal rank} if its rank is exactly $n(p)$.
\end{definition} 

Note that for protocols $p$ with all nodes unary, the sets $s(p)$ remain intrinsically linearly ordered: for $p=\branching{x}{v}{q}$ or $p=\multiplicative{x}{r}{q}$, the site $s(x)$ precedes those in $s(q)$.  This intrinsic order on $s(p)$ extends to $s(p)+1$ by placing the additional output point last; we use this extended order throughout.

By construction, $T_p$ is strictly lower triangular in the intrinsic order of $p$.
As $T_p$ is not a square matrix, this needs further clarification. The rows of $T_p$ are indexed by $U+V+1$ and the columns are indexed by $V\subseteq U+V+1$.
By strict lower triangularity we mean that for each $i\in V$, all entries on the $i$th column up to $i$ in the intrinsic order are zero; explicitly, $T_{j,i} = 0$ for all $i \in V$ and $j \leq i$ in this order. 

The operational meaning of $T_p$ is the same as in \Cref{ssec:linearprotocols} for linear protocols:
the column indexed by branching site $i \in V$ specifies how the outcome at site $i$ affects the measurement settings at other sites.
Strict lower triangularity ensures that a measurement outcome can only influence the settings of strictly later measurements.

\begin{remark}\label{rem:reconstructingMPfromamatrix}
As in \Cref{rem:reconstructinglinearMPfromamatrix}, specifying a $(U+V+1)\times V$-matrix $T$ that is strictly lower triangular in some ordering of $U+V+1$ (with the output point last) does not uniquely determine a protocol $p\in \MP_{U,V}$ with unary nodes:
the same matrix may be strictly lower triangular in several orderings, and additional data is required even once an ordering is fixed. Specifically, to reconstruct $p$, we must also provide
a vector $w_p \in \ZZ_2^{s(p)+1}$ of default measurement settings (and default output bit) and a vector $r_p \in \ZZ_2^U$ of required outcomes for the multiplicative nodes. Together with the ordering and the adaptivity matrix $T_p$, these data determine a unique measurement protocol $p\in \MP_{U,V}$. As before, the action $(v,p)\mapsto v.p$ translates in this representation to $v.(T_p,w_p,r_p)=(T_p,v+w_p,r_p)$.
\end{remark}

\subsection{The scenario of measurement protocols}

We now want to define the scenario whose measurements are the measurement protocols.
For this, we require an auxiliary notion to define compatibility for measurement protocols.
Intuitively, $W(p,\gs)$ is the set of measurements that are actually performed when running a protocol $p$ with outcomes determined by a global assignment $\gs$, following the path through the tree that $\gs$ selects at each node. In parallel, we define $o(p,\gs)\in\ZZ_2$ to record the output bit produced by that run.

\begin{definition}\label{def:Wando}
    Let $p$ be a measurement protocol on the scenario $\bell$ and
    $\gs \colon \coprod_{i\in \II} \ZZ_2\to\ZZ_2$ a global assignment of outcomes to the measurements of $\bell$.
    We define the set $W(p,\gs)$ of \emph{measurements performed by $p$ given $\gs$} and the \emph{output of $p$ given $\gs$}, $o(p,\gs)\in\ZZ_2$,
    by simultaneous induction on the structure of $p$:
    \begin{itemize}
    \item  If $p \in \MP_{\emptyset,\emptyset}$,
           set $W(p,\gs)=\emptyset$ and $o(p,\gs)=p$ (viewing $p\in\ZZ_2$).
    \item If $p=\branching{D}{v}{q}$ starts with a branching node, 
		    \begin{align*}
                W(p,\gs) &= \begin{cases}
			                    D \cup W(q,\gs)   & \text{ if $\sum_{x\in D}\gs(x)=0$,}\\
				                  D \cup W(v.q,\gs) & \text{ otherwise;}
		                      \end{cases}\\[4pt]
                o(p,\gs) &= \begin{cases}
			                    o(q,\gs)   & \text{ if $\sum_{x\in D}\gs(x)=0$,}\\
                                o(v.q,\gs) & \text{ otherwise.}
		                 \end{cases}
            \end{align*}
    \item If $p=\multiplicative{D}{r}{q}$ starts with a multiplicative node, 
            \begin{align*}
                W(p,\gs) &= \begin{cases}
			                    D \cup W(q,\gs)  & \text{ if $\sum_{x \in D} \gs(x)=r$,}\\
				                D  & \text{ otherwise.}
		                      \end{cases}\\[4pt]
                o(p,\gs) &= \begin{cases}
			                    o(q,\gs)   & \text{ if $\sum_{x \in D} \gs(x)=r$,}\\
				                0 & \text{ otherwise.}
		                    \end{cases}
            \end{align*}
    \end{itemize}
\end{definition}
Note that $o(p,\gs)$ depends on $\gs$ only through $\gs|_{W(p,\gs)}$, \ie the output is determined by the outcomes of the measurements actually performed in the run. By the admissibility requirements in \Cref{def:moregeneralmps}, the set $W(p,\gs)$ is always a context of $\bell$.

\begin{definition}\label{def:MPScenario}
Given the Bell scenario $\bell$ on a set of sites $\II$,
we define the \emph{scenario $\MP(\bell,\ZZ_2)$ of measurement protocols on $\bell$}:
\begin{itemize}
   \item The measurements are given by measurement protocols on $\bell$ (in the sense of \Cref{def:moregeneralmps}).
   \item A set $P$ of measurement protocols forms a \emph{context} if, for every global assignment $\gs$ of outcomes to the measurements of $\bell$,
   the set $W(P,\gs) \defeq\bigcup_{p\in P} W(p,\gs)$ is a context of $\bell$.
   \item The outcome set is fixed to be $\ZZ_2$.
\end{itemize}
\end{definition}

Every empirical model $e \colon \bell$ extends to an empirical model $\MP(e,\ZZ_2) \colon \MP(\bell,\ZZ_2)$,
defined by assigning to each context $P$ the distribution over joint outputs obtained by jointly running each protocol $p \in P$ on $e$ and recording the output bit. We briefly set out its definition below;
see~\cite[Definition~16]{abramskyetal:comonadicview} for a precise definition with a more detailed explanation in an analogous setting --
the main difference being that there the full sequence of outcomes observed throughout a run is recorded, rather than a single output bit; the model we define here is thus a coarse-graining thereof.

Recall from \Cref{sec:background} the definition of the event sheaf: $\ES(U)$ denotes the set of outcome assignments to the measurements in a context $U$. For the scenario $\bell$, $\ES(U)$ is the set of functions $U\to\ZZ_2$, where $U$ is a subset of the measurement set $X=\coprod_{i\in\II}\ZZ_2$ of $\bell$.

Given a context $P$ of $\MP(\bell,\ZZ_2)$, write
	\[\FF_P \defeq
	\{ \gs|_{W(P,\gs)} \in \ES(W(P,\gs))  \mid \gs \in \ES(X) \}
	\]
for the set of partial assignments of outcomes to measurements that may be observed in a (joint) run of the protocols in $P$.
Since, as remarked right after its definition, $o(p,\gs)$ depends on $\gs$ only through $\gs|_{W(p,\gs)}$,
each map
	$o(p,\cdot) \colon \ES(X) \to \ZZ_2$ 
factors through the restriction $\gs\mapsto\gs\vert_{W(P,\gs)}$,
yielding a well-defined map $\tilde{o}(p,\cdot)\colon\FF_P\to\ZZ_2$.

\begin{definition}\label{def:MPmodel}
	Given an empirical model $e \colon \bell$, the induced empirical model $\MP(e,\ZZ_2) \colon \MP(\bell,\ZZ_2)$
	is defined as follows: for each context $P$ of $\MP(\bell,\ZZ_2)$, the probability of observing an outcome assignment $r\colon P \to \ZZ_2$ is
	\[
		\MP(e,\ZZ_2)_P(r) = \sum_{\substack{(\ls \colon U \to \ZZ_2) \in \FF_P \\ \forall p \in P,\, \tilde{o}(p,\ls)=r(p)}}e_U(\ls)~.
	\]
\end{definition}

In much of what follows we prove equations that hold in the $\ZZ_2$-linear theory of $\MP(e,\ZZ_2)$ for \emph{every} empirical model $e\colon\bell$,
\ie equations in the intersection of the $\ZZ_2$-linear theories of $\MP(e,\ZZ_2)$ over all $e\colon\bell$.
This common theory may be thought of as the $\ZZ_2$-linear theory of the scenario $\MP(\bell,\ZZ_2)$ itself.
Of course, the theory common to all empirical models on any given measurement scenario is trivial, consisting only of tautologies.
Here it is not, because the models of the form $\MP(e,\ZZ_2)$ for $e\colon\bell$ constitute a proper subclass of all the empirical models on $\MP(\bell,\ZZ_2)$ seen as a bare measurement scenario:
it contains only those arising by \emph{running} measurement protocols on some underlying $\bell$-model.
Thus the $\MP$ construction implicitly constrains the class of empirical models deemed valid in the scenario.
The following proposition characterises this common theory, showing that validity of an equation reduces to a condition on global outcome assignments.

\begin{proposition}\label{prop:allMPequations}
  Let $P$ be a context of $\MP(\bell,\ZZ_2)$ and $a\in\ZZ_2$.
  The following are equivalent:
  \begin{enumerate}[label=(\roman*)]
	  \item $\MP(e,\ZZ_2)\vDash\sum_{p\in P}p=a$ for every empirical model $e\colon\bell$;
    	  \item for every global assignment $\gs\colon\coprod_{i\in\II}\ZZ_2\to\ZZ_2$, $\sum_{p\in P}o(p,\gs)=a$.
  \end{enumerate}
\end{proposition}
\begin{proof}
	($\Rightarrow$)\enspace
	Suppose $\MP(e,\ZZ_2)\vDash\sum_{p\in P}p=a$ for every $e\colon\bell$.
	In particular, any global assignment $\gs$ determines a (deterministic, non-contextual) empirical model $\delta_\gs\colon\bell$ defined by $(\delta_\gs)_C\defeq\delta_{\gs|_C}$, which produces outcomes according to $\gs$ with certainty.
    Such a deterministic model lifts to a deterministic model on $\MP(\bell,\ZZ_2)$ given by $\MP(\delta_\gs,\ZZ_2) = \delta_{o(\cdot,\gs)}$,
    with $o(\cdot,\gs)$ being a global assignment for the scenario $\MP(\bell,\ZZ_2)$.
	The equation holding for $\delta_\gs$ thus gives $\sum_{p\in P}o(p,\gs)=a$.

	($\Leftarrow$)\enspace
	Suppose $\sum_{p\in P}o(p,\gs)=a$ for every $\gs$, and let $e\colon\bell$ be an arbitrary empirical model.
	An outcome assignment $r\colon P\to\ZZ_2$ in the support of $\MP(e,\ZZ_2)_P$
	arises from some $\ls\in\FF_P$ with $\tilde{o}(p,\ls)=r(p)$ for all $p\in P$.
	In turn, such $\ls$ is the restriction of some global assignment $\gs$ with $W(P,\gs)=\dom(\ls)$
	(one may freely assign outcomes to measurements outside $\dom(\ls)$ to obtain such an extension), so that $o(p,\gs)=\tilde{o}(p,\ls)=r(p)$ for all $p\in P$.
	By assumption, $\sum_{p\in P}o(p,\gs)=a$, hence $\sum_{p\in P}r(p)=a$.
\end{proof}

Both the construction of the induced model $\MP(e,\ZZ_2)$ (\Cref{def:MPmodel}) and the characterisation of the theory common to all such models on a given scenario (\Cref{prop:allMPequations}) extend verbatim to
arbitrary measurement protocols in the broader sense considered in~\cite{abramskyetal:comonadicview}, where the possible continuations of a branching node may be arbitrary, rather than linearly related.
Specifically for the measurement protocols considered in this paper -- where adaptivity is \emph{linear}, in that the outcome of a branching node acts on the continuation by adding a vector to later measurement settings -- the following proposition gives a more concrete operational description using the adaptivity matrix.
We restrict to protocols with all nodes unary, since the adaptivity matrix has only been defined in that setting; this is a technical convenience that keeps the exposition simple, and could be lifted at the cost of some extra complexity.
In the statement below, the vector $b$ records the measurement outcomes that are (or would be) observed at each branching site, agreeing with the value assigned by $\gs|_{W(p,\gs)}$ above.

\begin{proposition}\label{prop:runviamatrix}
Let $p\in\MP_{U,V}$ be a  protocol with all nodes unary, with adaptivity matrix $T=T_p$, default measurement settings and output  $w=w_p\in\ZZ_2^{s(p)+1}$ and required outcomes  $r=r_p\in\ZZ_2^U$ for the multiplicative sites. Given a global assignment $\gs\colon\coprod_{i\in\II}\ZZ_2\to\ZZ_2$, define
the vector $b=b(p,\gs)\in\ZZ_2^V$ of measurement outcomes at the branching sites by recursion along the intrinsic order: for each branching site $j \in V$,
\[ b_j \defeq \gs\big(x_{j,\,(w+Tb)_j}\big)
          = \gs\big(x_{j,\;w_j+\sum_{i\in V,\,i<j}T_{j,i}b_i}\big),\]
which is well defined since $T$ is strictly lower triangular.

The run of $p$ against $\gs$ measures $x_{j,(w+Tb)_j}$ at each visited site $j$.  
If $\gs(x_{k,(w+Tb)_k}) = r_k$ for every multiplicative site $k \in U$, then the run completes without aborting:
\[ W(p,\gs)=\{\,x_{j,(w+Tb)_j}\mid j\in s(p)\,\},\qquad o(p,\gs)=(w+Tb)_\ast,\]
with $\ast$ the output point.
Otherwise, letting $k^\ast$ be the least $k\in U$ such that $\gs\big(x_{k,(w+Tb)_k}\big)\neq r_k$, the run aborts at $k^\ast$:
\[ W(p,\gs)=\{\,x_{j,(w+Tb)_j}\mid j\in s(p),\ j\le k^\ast\,\},\qquad o(p,\gs)=0.\]
\end{proposition}
\begin{proof}
    The proof is straightforward, by induction on the structure of $p$, using the definitions of $T_p$ (\Cref{def:adaptivitymatrix}) and of $W(p,\gs)$ and $o(p,\gs)$ (\Cref{def:Wando}). 
\end{proof}

\subsection{The main inductive argument}

As noted in \Cref{ssec:computingwithlinearprotocols}, to prove \Cref{thm:main}, it suffices to establish equation~\eqref{eq:sumstozero}. We now state the latter more precisely. 

\begin{theorem}\label{thm:sumstozero} Let $p\in\MP_{U,V}$ be a measurement protocol on $\bell$ with all nodes unary, and $e: \bell$ be an empirical model.
If $Q\colon \ZZ_2^{l}\to \ZZ_2^{s(p)+1}$ is a linear map where $l\geq m(p)+2$, then 
\begin{equation}\label{eq:sumstozerov2}
\MP(e,\ZZ_2)\vDash^*\sum_{w\in \ZZ_2^{l}} Q(w).p=0~.
\end{equation}
\end{theorem}

Ultimately we care about the case when $m(p)=0$:
the added generality afforded by multiplicative nodes
is not of intrinsic interest here, but is needed to ensure a sufficiently strong induction hypothesis.

We now give the overall structure of the proof, formalising the proof sketch of \Cref{ssec:computingwithlinearprotocols}; the key lemmas are stated and proved in \Cref{sec:toolsfortheproof}.
\begin{proof}[Proof of \Cref{thm:sumstozero}]
We prove the claim by induction on the number of branching levels $n(p)$.

If $n(p)=0$, then the codomain of $Q$ has dimension $m(p)+n(p)+1=m(p)+1<l$, so $Q$ is not injective. As a result, each summand in \cref{eq:sumstozerov2} appears an even number of times and hence cancels out, so that the equation holds.

For $n(p)>0$, if $Q$ fails to be injective,  \cref{eq:sumstozerov2} holds by the argument above. If $Q$ is injective, the result follows from the following three lemmas: 
\begin{enumerate}
	\item If $\rank T_p=n(p)$, \Cref{lem:cansplice} shows that there is another measurement protocol $q\in \MP_{U,V}$ with all nodes unary such that
		\[\sum_{w\in \ZZ_2^{l}} Q(w).p \;=\, \sum_{w\in \ZZ_2^{l}} Q(w).q\]
	and $\rank T_q<n(p)=n(q)$. In other words, without loss of generality, we can assume that $T_p$ is not of maximal rank.
	\item If $\rank T_p<n(p)$, \Cref{lem:nonmaximalrankacrossthesum} shows that we can condense $p$ to have strictly fewer branching levels without affecting the value of the sum $\sum_w Q(w).p$, at the cost of increasing the arity of some nodes.
	\item Given a measurement protocol, \Cref{lem:sumofsmallersums} shows that we can systematically replace all higher-arity branching levels by sums of protocols with unary nodes without increasing the number of branching levels.  Combined with Step 2, this lets us write the overall sum as a sum of smaller sums $\sum_v Q_i(v).p_i$, where each $p_i$ has strictly fewer than $n(p)$ branching levels, therefore reducing to the induction hypothesis. \qedhere
\end{enumerate}
\end{proof}

\section{Tools for the proof}\label{sec:toolsfortheproof}

In this section, we establish the key lemmas on which the proof of \Cref{thm:sumstozero} rests.
After a preliminary result on decomposing protocols into branches (\Cref{ssec:branches}), each of the three main subsections addresses one step of the proof:
splicing to reduce the rank of the adaptivity matrix (\Cref{ssec:splicing}), condensing protocols with non-maximal rank into fewer branching levels (\Cref{ssec:condensing}), and removing higher-arity branching nodes by expressing them as sums of simpler protocols (\Cref{ssec:higherarity}).
Each of these three follows a common pattern: a key construction is introduced and a general lemma about it established, which is then applied in the specific context of \Cref{thm:sumstozero}.

\subsection{Decomposing into branches}
\label{ssec:branches}
We begin by establishing that every measurement protocol decomposes as the sum of its branches, a fact used throughout the sequel.

Recall from \Cref{def:moregeneralmps} that a branch is a measurement protocol with only multiplicative nodes, ending in $1$. It records a single deterministic path through a protocol tree, given by a sequence of parity measurements each with a required outcome, that returns $1$ when every one of them is observed, and $0$ otherwise.

\begin{definition}\label{def:branches} For a measurement protocol $p$, we define its set of branches, $\Branches(p)$, by induction on the structure of $p$:
\begin{itemize}
	\item For $p \in \MP_{\emptyset,\emptyset}$, we set $\Branches(0)=\emptyset$ and $\Branches(1)=\{1\}$. 
	\item If $p=\multiplicative{D}{r}{q}$ starts with a multiplicative node, then 
		\[\Branches(p)\defeq \left\{ \multiplicative{D}{r}{t} \mid t\in \Branches(q) \right\}.\]
	\item If $p = \branching{D}{v}{q}$ starts with a branching node, then 
            \[\Branches(p) \defeq \Branches(\multiplicative{D}{0}{q}) \cup \Branches(\multiplicative{D}{1}{v.q}).\]
\end{itemize}
\end{definition}

\begin{lemma}\label{lem:eqsarecongruent}
The $\ZZ_2$-linear theory of $\MP(e,\ZZ_2)$ respects pre-composition by a multiplicative node: if 
$\sum_{i} p_i = \sum_{j} q_j$
and each 
\[\multiplicative{D}{r}{p_i} = \vcenter{\hbox{\begin{forest} [$x_1+\dots + x_d$ [$p_i$, edge label={node[midway,right]{$r$}} ]]\end{forest}}}
\quad \text{and} \quad
\multiplicative{D}{r}{q_j} = \vcenter{\hbox{\begin{forest} [$x_1+\dots + x_d$ [$q_j$, edge label={node[midway,right]{$r$}} ]]\end{forest}}}
\]
is a well-defined protocol,
then $\sum_i \multiplicative{D}{r}{p_i} = \sum_j \multiplicative{D}{r}{q_j}$, \ie
\[
\sum_{i}\vcenter{\hbox{\begin{forest} [$x_1+\dots + x_d$ [$p_i$, edge label={node[midway,right]{$r$}} ]]\end{forest}}}
 \;\;=\;\; 
\sum_{j} \vcenter{\hbox{\begin{forest} [$x_1+\dots + x_d$ [$q_j$, edge label={node[midway,right]{$r$}} ]]\end{forest}}}. 
\]
\end{lemma}
\begin{proof}
     This follows from \Cref{prop:allMPequations} by a straightforward case split, for each global assignment $\gs$, on whether $\sum_{x\in D}\gs(x)=r$.
\end{proof}

\begin{lemma}\label{lem:sumofbranches} Every measurement protocol $p$ equals the sum of its branches:
\begin{equation}\label{eq:sumofbranches} p=\sum_{s\in \Branches(p)} s~.
\end{equation}
\end{lemma}
\begin{proof} We prove this by induction on the structure of $p$.
\begin{itemize}
    \item If $p \in \MP_{\emptyset,\emptyset}$, the claim is clear.
    \item If $p=\multiplicative{D}{r}{q}$ starts with a multiplicative node, where $D=\{x_1,\dots,x_d\}$, then 
\[\vcenter{\hbox{\begin{forest} [$x_1+\dots + x_d$ [$q$, edge label={node[midway,right]{$r$}} ]]\end{forest}}}=\sum_{t\in \Branches(q)} \vcenter{\hbox{\begin{forest} [$x_1+\dots + x_d$ [$t$, edge label={node[midway,right]{$r$}} ]]\end{forest}}} = \sum_{s\in \Branches(p)} s \]
 where the first equality follows from the induction hypothesis by \Cref{lem:eqsarecongruent} and the second by the multiplicative-node case of the definition of $\Branches(p)$ (\Cref{def:branches}).
 \item If $p = \branching{D}{v}{q}$ starts with a branching node, where $D=\{x_1,\dots,x_d\}$, then 
\begin{align*}
&\vcenter{\hbox{\begin{forest}
[     $x_1+\dots + x_d$
[$q$, edge label={node[midway,left]{$0$}}]		[$v.q$,  edge label={node[midway,right]{$1$}}]] \end{forest}}} =
\vcenter{\hbox{\begin{forest} [$x_1+\dots + x_d$ [$q$, edge label={node[midway,right]{$0$}} ]]\end{forest}}} + \vcenter{\hbox{\begin{forest} [$x_1+\dots + x_d$ [$v.q$, edge label={node[midway,right]{$1$}} ]]\end{forest}}} 
\\
&=\sum_{t\in \Branches(q)} \vcenter{\hbox{\begin{forest} [$x_1+\dots + x_d$ [$t$, edge label={node[midway,right]{$0$}} ]]\end{forest}}} + \sum_{t\in \Branches(v.q)} \vcenter{\hbox{\begin{forest} [$x_1+\dots + x_d$ [$t$, edge label={node[midway,right]{$1$}} ]]\end{forest}}} = \sum_{s\in \Branches(p)} s
\end{align*}
where the first equality holds for all models on $\bell$, because both sides measure $x_1+\dots +x_d$ and then depending on the outcome return either $q$ or $v.q$,  the second follows from the induction hypothesis by \Cref{lem:eqsarecongruent}, and the final one by the branching-node case of the definition of $\Branches(p)$ (\Cref{def:branches}).
\end{itemize}
\end{proof}

\subsection{Splicing}
\label{ssec:splicing}

We now address Step 1 of the proof of \Cref{thm:sumstozero}: given a protocol whose adaptivity matrix has maximal rank, we produce an equivalent one of strictly lower rank. We achieve this via a \emph{splicing} construction, which operates on a pair of related protocols, rearranging their branches to modify the adaptivity structure while preserving their sum.

The main construction of this subsection operates on a measurement protocol $p$ with unary nodes and its translate $v.p$ by a vector $v$ with sufficiently many leading zeroes.
Let $x$ be the measurement at any branching node of $p$ that precedes the first nonzero entry of $v$ (in the intrinsic order of $p$). Since $v$ has leading zeroes up to and including $s(x)$,  $p$ and $v.p$ agree up to $x$, diverging only in the continuations that follow it. Thus, schematically, these two protocols look like
\[p=\vcenter{\hbox{\begin{forest}[$p_{_1}$ [ $x$  [$p_{_2}$, edge label={node[midway,left]{$0$}}]	[$v'.p_{_2}$,  edge label={node[midway,right]{$1$}}]]]\end{forest}}} 
\enspace \text{and} \enspace 
v.p = \!\!\!\!\!
\vcenter{\hbox{\begin{forest}[$p_{_1}$ [ $x$  [$\pi_{s(p_2)+1}(v).p_{_2}$, edge label={node[midway,left,yshift=2.5pt,xshift=1pt]{$0$}}]	[$(\pi_{s(p_2)+1}(v) + v').p_{_2}$,  edge label={node[midway,right,yshift=2.5pt,xshift=-1pt]{$1$}}]]]\end{forest}}}\] 
where $p_1$ denotes an initial subprotocol on which $v$ acts trivially.

We split both protocols along $x$ and recombine their halves, obtaining the two protocols below. Denoting the first of these by $q$, the second is then $v.q$:
\[q \defeq \vcenter{\hbox{\begin{forest}[$p_{_1}$ [ $x$  [$p_{_2}$, edge label={node[midway,left,yshift=4pt,xshift=1pt]{$0$}}]	[$(\pi_{s(p_{2})+1}(v) + v').p_{_2}$,  edge label={node[midway,right,yshift=2.5pt,xshift=-1pt]{$1$}}]]]\end{forest}}}  
\enspace \text{and} \enspace
v.q = \vcenter{\hbox{\begin{forest}[$p_{_1}$ [ $x$  [$\pi_{s(p_{2})+1}(v).p_{_2}$, edge label={node[midway,left,yshift=2.5pt,xshift=1pt]{$0$}}]	[$v'.p_{_2}$,  edge label={node[midway,right,yshift=4pt,xshift=-1pt]{$1$}}]]]\end{forest}}}\] 
We call these two protocols the result of \emph{splicing $p$ and $v.p$ along $x$}.\footnote{While we could give a detailed inductive definition and accompanying proofs for the results in this subsection, this would be notationally heavy. We believe that the slightly less formal exposition here is more illuminating.}

The multiset of branches across the two resulting protocols remains the same, so their total sum is preserved.
Moreover, the spliced pair of protocols $q$ and $v.q$  still differ by the same vector $v$ as the original pair $p$ and $v.p$.
However, while in $p$ and $v.p$ the outcome of $x$ affects later settings via $v'$, in $q$ and $v.q$, it does so via $v'+\pi_{s(p_2)+1}(v)$.
Choosing a suitable splicing thus lets us rewrite $\sum_w Q(w).p$ as $\sum_w Q(w).q$ with a precise description of how the adaptivity matrix $T_q$ differs from $T_p$.
We collect this into the following statement.

\begin{lemma}\label{lem:splicing}
Let $p$ be a measurement protocol with all nodes unary, $x$ a branching node of $p$, and $v \in \ZZ_2^{s(p)+1}$ such that $\pi_i(v) = 0$ for all sites $i \leq s(x)$ in the intrinsic order of $p$. Then splicing $p$ and $v.p$ along $x$ results in protocols $q$ and $v.q$ satisfying
\begin{enumerate}
    \item\label{lem:splicing-item1} $p+v.p=q+v.q$.
    \item\label{lem:splicing-item2} The  adaptivity matrix $T_q$ differs from $T_p$ only in column $s(x)$, where it equals the corresponding column of $T_p$ plus $v$. Equivalently, for $u \in \ZZ_2^V$,  $T_q(u)=T_p(u)+\pi_{s(x)}(u)\,v$. 
\end{enumerate}
\end{lemma}

We apply \Cref{lem:splicing} in the context of \Cref{thm:sumstozero}.
The following lemma implements Step 1.
A suitable splicing is found via a dimension count: the images of $T_p$ and $Q$ together exceed the ambient dimension, so they must intersect nontrivially, yielding the required vector $v$.

\begin{lemma}\label{lem:cansplice}
Let $p\in \MP_{U,V}$ be a measurement protocol with all nodes unary and $T_p$ of maximal rank,
and let $Q\colon \ZZ_2^{l}\to \ZZ_2^{s(p)+1}$ be an injective linear map where $l\geq m(p)+2$.
Then there exists a measurement protocol $q\in\MP_{U,V}$  
with $n(q)=n(p)$ such that $\rank(T_q)<n(q)$
and
 \[\sum_{w\in \ZZ_2^{l}} Q(w).p \;=\, \sum_{w\in \ZZ_2^{l}} Q(w).q.\]
\end{lemma}

\begin{proof} The space $\ZZ_2^{s(p)+1}$ has dimension $m(p)+n(p)+1$.
By the maximal rank assumption, $\im T_p$ is a subspace of dimension $n(p)$ while since $Q$ is injective, $\im Q$ is a subspace of dimension $l\geq m(p)+2$.
Since \[\dim(\im T_p) + \dim(\im Q) = n(p)+l\geq m(p)+n(p)+2>\dim(\ZZ_2^{s(p)+1}),\]
these subspaces must intersect nontrivially, so we can find $u^*\in \ZZ_2^V$ and $w^*\in \ZZ_2^{l}$ such that $T_p u^*=Q w^*\neq 0$. Let $x$ be the measurement corresponding to the first (in the intrinsic order) nonzero component of $u^*$, which must exist since $T_p(u^*)\neq 0$ and so $u^*\neq 0$.
By strict lower triangularity of $T_p$, the first nonzero component of $T_p(u^*)=Q(w^*)$ lies strictly after $s(x)$. Hence, $p$ and $Q(w^*).p$ can be spliced along $x$, resulting in protocols $q$ and $Q(w^*).q$ as in \Cref{lem:splicing}. 

To show that \[\sum_{w\in \ZZ_2^{l}} Q(w).p \;=\, \sum_{w\in \ZZ_2^{l}} Q(w).q , \]
express $\ZZ_2^{l}$ as $\ZZ_2 w^*\oplus W$. Then 
 \begin{align*}\sum_{w\in \ZZ_2^{l}} Q(w).p \;&=\,\sum_{w\in W}Q(w).p+Q(w+w^*).p \\ 
 &=\,\sum_{w\in W}Q(w). (p+ Q(w^*).p) \\
 &=\,\sum_{w\in W}Q(w). (q+ Q(w^*).q)\\ 
 &=\,\sum_{w\in W}  Q(w).q+Q(w+w^*).q \\
 &=\,\sum_{w\in \ZZ_2^{l}} Q(w).q\end{align*}
where we used \Cref{lem:splicing}--\ref{lem:splicing-item1} when moving to the third line.

To see that $\rank(T_q)<n(q)$, it suffices to exhibit a nonzero vector in the kernel of $T_q$.
As $T_q$ differs from $T_p$ only at the column of $s(x)$, the leading $1$ of $u^*$, \Cref{lem:splicing}--\ref{lem:splicing-item2}
gives us $T_q(u^*)=T_p(u^*)+\pi_{s(x)}(u^*)\, Q(w^*) = T_p(u^*)+ Q(w^*) = 0$
as desired. 
\end{proof}

\subsection{Condensing}
\label{ssec:condensing}

We now address Step 2 of the proof of \Cref{thm:sumstozero}: given a protocol not of maximal rank, we condense it into an equivalent one with strictly fewer branching levels, at the cost of introducing higher-arity nodes. The key point is that a non-maximal-rank adaptivity matrix effectively encodes fewer independent branching decisions than its depth suggests; we exploit this redundancy to condense it into a shallower equivalent.

\begin{lemma}\label{lem:nonmaximalrank} Let $p\in \MP_{U,V}$ be a measurement protocol with all nodes unary. Then there exists a measurement protocol $q\in \MP_{U,V'}$ for some $V'\subseteq V$ that has $n(q)=\rank T_p$  branching levels and satisfies $\MP(e,\ZZ_2)\vDash q=p$ for any $e:\bell$.
\end{lemma} 

\begin{proof}
The construction relies on choosing particular rows from the adaptivity matrix, resulting in a convenient basis for the row space. Each row in the basis corresponds to a branching level, determining the sum of measurements to be performed; multiplicative levels are interleaved wherever they appear, with the choice of the basis ensuring that the resulting protocol is well defined. In a bit more detail, we build the protocol by going through each row of the matrix from bottom to top. For each row, what we do depends on (i) whether the row in question was included in the basis and (ii) whether the row in question is for a multiplicative or a branching site. When the row in question was included in the basis, we add a corresponding branching node. When the row in question was for a multiplicative node, we add the same multiplicative node. When both conditions hold, we do both constructions, with the branching node before the multiplicative one. Let us now fill in the details.  

Let $T$ be the  $(U+V+1)\times V$ adaptivity matrix of $p$, $w\in\ZZ_2^{s(p)+1}$ the vector of default measurement settings of $p$, and $r\in\ZZ_2^U$ the vector of required outcomes for the multiplicative nodes. For $j\in U+V+1$, write $t_j$ for the $j$th row of $T$. Call a row $t_j$ a pivot if $t_j$ is not in the span of the set $\{t_i \mid i<j\}$ of strictly earlier rows. The pivot rows form a basis of the row space with the property that every row $t_j$ can be uniquely written as a sum $t_j=\sum_i c_{j,i}t_i$ of pivot rows weakly before the $j$th site so that $c_{j,i}=0$ whenever $i>j$.
For a row $t_j$, let $D_j$ be the set of measurements at the sites indexing the non-zero entries of $t_j$, with default measurement settings determined by $w$, \ie $D_j = \{x_{i,w_i} \mid T_{j,i} \neq 0\}.$

Let $1,\dots ,N$ enumerate $s(p)+1$ in the intrinsic order, with $N=\ast$ the output point.
We construct protocols $q_N,\dots q_1$ inductively, with $q\defeq q_1$. We first set 
\[ q_N \defeq
   \begin{cases}
     \branching{D_N}{1}{w_N} & \text{if $t_N$ is a pivot row,}\\
     w_N & \text{otherwise.}
   \end{cases} \]
For $i<N$, we  first define $v_i\in \ZZ_2^{s(q_{i+1})+1}$ by $v_i\defeq (c_{j,i})_{j\in s(q_{i+1})+1}$ and then set
\[q_{i}\defeq
   \begin{cases}
     \hfil\vcenter{\hbox{\begin{forest}
       [ $D_i$
         [$q_{i+1}$, edge label={node[midway,left]{$0$}}]
         [$v_i.q_{i+1}$, edge label={node[midway,right]{$1$}}]]
     \end{forest}}}
       & \text{if $t_{i}$ is a pivot row and $i$ is a branching site.}\\[12pt]
     \hfil\vcenter{\hbox{\begin{forest}
       [ $D_i$
         [$x_{i,w_i}$, edge label={node[midway,left]{$0$}}
           [$q_{i+1}$, edge label={node[midway,left]{$r_i$}}]]
         [$1.x_{i,w_i}$, edge label={node[midway,right]{$1$}}
           [$v_i.q_{i+1}$, edge label={node[midway,right]{$r_i$}}]]]
     \end{forest}}}
       & \text{if $t_{i}$ is a pivot row and $i$ is a multiplicative site.}\\[12pt]
     \hfil q_{i+1} & \text{if $t_{i}$ is not a pivot row and $i$ is a branching site.}\\[6pt]
     \hfil\vcenter{\hbox{\begin{forest}
       [ $x_{i,w_i}$ [$q_{i+1}$, edge label={node[midway,right]{$r_i$}}]]
     \end{forest}}}
       & \text{if $t_{i}$ is not a pivot row and $i$ is a multiplicative site.}
   \end{cases}
\]

We first argue that each $q_i$ is a well-formed measurement protocol. By construction, the branching nodes correspond to basis vectors and hence operate on pairwise distinct sets of sites. The multiplicative nodes are already multiplicative nodes in $p$, and hence pairwise distinct. Moreover, $U$ and $V$ are disjoint (as all nodes in $p$ were unary), so that all branching nodes of each $q_i$ are also distinct from the multiplicative nodes.  Furthermore, $c_{j,i}=0$ whenever $i>j$ in the intrinsic order, and lower triangularity of $T$ implies that if $x\in D_i$ for some $i$, then $s(x)<i$.
Therefore, the coordinates of later action vectors $v_j$ for $j\geq i$ vanish on $s(D_i)$, so that every occurrence in $q_i$ of a site in $D_i$ is with the same measurement setting as in $D_i$. This ensures admissibility, so that each $q_i$ is indeed a measurement protocol. By construction  $q\in \MP_{U,V'}$ for some $V'\subseteq V$, and $q$ has precisely $n(q)=\rank T$  branching levels as these correspond to a basis of the row space of $T$. 

To see that $\MP(e,\ZZ_2)\vDash q=p$ for any $e: \bell$, by \Cref{prop:allMPequations} it suffices to show that for all global assignments on $\bell$, both $q$ and $p$ result in the same outcome.
For that purpose, fix a global assignment $\gs \colon \coprod_{i\in \II} \ZZ_2\to\ZZ_2$. 
Let $b\in \ZZ_2^V$ be defined as in \Cref{prop:runviamatrix}, which applies to $p$. While the proposition as stated does not apply to $q$ (which has non-unary nodes and hence we haven't defined its matrix), $q$ behaves similarly to $p$ in that, 
barring them aborting early, the vector giving the measurement settings and the final outcome bit of both $p$ and $q$ is given by $w+Tb$. Moreover, as $p$ and $q$ have the same multiplicative nodes with the same expected outcomes, $p$ aborts iff $q$ aborts, in which case they both return $0$. Thus  $W(q,\gs)\subseteq W(p,\gs)$ and in particular $p$ and $q$ are compatible. Finally, as the final output bit of both $p$ and $q$ is given by $(w+Tb)_{\ast}$ when they don't abort, we have concluded that $p$ and $q$ return the same outcome for every global assignment, so that $\MP(e,\ZZ_2)\vDash q=p$. 
\end{proof}

We apply \Cref{lem:nonmaximalrank} in the context of \Cref{thm:sumstozero}. The following lemma implements Step 2, condensing each protocol in the sum uniformly along the action of $Q$.

\begin{lemma}\label{lem:nonmaximalrankacrossthesum} Let $p$ be a measurement protocol with all nodes unary and $Q\colon \ZZ_2^{l}\to \ZZ_2^{s(p)+1}$ be a linear map where $l\geq m(p)+2$. Denote by $q$ the protocol produced by \Cref{lem:nonmaximalrank}. Then  
\[\sum_{w\in \ZZ_2^{l}} Q(w).p \;=\, \sum_{w\in \ZZ_2^{l}} (\pi_{s(q)+1}\circ Q)(w).q~.\]
\end{lemma}

\begin{proof}
As the construction in \Cref{lem:nonmaximalrank} reads only the matrix and the order, with the default measurement settings carried along, it commutes with the action of the input in the sense that $Q(w).p = (\pi_{s(q)+1}\circ Q)(w).q$,
giving the result. 
\end{proof}

\subsection{Removing higher-arity nodes}
\label{ssec:higherarity}

We now address Step 3 of the proof of \Cref{thm:sumstozero}: eliminating higher-arity branching nodes by expressing $\sum_w Q(w).p$ as a sum of smaller sums over protocols with all nodes unary. 
Each such reduction introduces new multiplicative nodes while correspondingly increasing the input dimension, so that the bound $l \geq m(p)+2$ is maintained throughout.

\begin{lemma}\label{lem:separatingsums} For any measurement protocol of the form 
\[ \vcenter{\hbox{\begin{forest} [$p$ [$x_1+\dots +x_d$ [$q$, edge label={node[midway,left]{$0$}} ] [$v.q$, edge label={node[midway,right]{$1$}} ]]]\end{forest}}}\]
we have 
\begin{equation}\label{eq:separatingsums} 
\vcenter{\hbox{\begin{forest} [$p$ [$x_1+\dots +x_d$ [$q$, edge label={node[midway,left]{$0$}} ] [$v.q$, edge label={node[midway,right]{$1$}} ]]]\end{forest}}}
	= \vcenter{\hbox{\begin{forest} [$p$ [$x_1$ [$q$, edge label={node[midway,left]{$0$}} ] [$v.q$, edge label={node[midway,right]{$1$}} ] ]]\end{forest}}} + 
    \sum_{i=2}^d \left (
	\vcenter{\hbox{\begin{forest} [$p$ [$x_i$ [$q$, edge label={node[midway,right]{$1$}} ]]]\end{forest}}} +
	\vcenter{\hbox{\begin{forest} [$p$ [$x_i$ [$v.q$, edge label={node[midway,right]{$1$}} ]]]\end{forest}}} \right )~.
	\end{equation}
\end{lemma}
\begin{proof} 
We first break the protocol into the sum of its branches using \Cref{lem:sumofbranches}.
Any branch of this protocol is of the form
\[ \vcenter{\hbox{\begin{forest} [$s$ [$x_1+\dots +x_d$ [$t$, edge label={node[midway,right]{$r$}} ]]]\end{forest}}} \]
for some branches  $s \in \Branches(p)$ and $t \in \Branches(q)$ and a bit $r\in\ZZ_2$.

We first argue that the following equation holds for each branch:
\begin{equation}\label{eq:separatingsumsforbranches} \vcenter{\hbox{\begin{forest} [$s$ [$x_1+\dots +x_d$ [$t$, edge label={node[midway,right]{$r$}} ]]]\end{forest}}}
	= \vcenter{\hbox{\begin{forest} [$s$ [$x_1$ [$t$, edge label={node[midway,right]{$r$}} ]]]\end{forest}}} + 
	\vcenter{\hbox{\begin{forest} [$s$ [$x_2$ [$t$, edge label={node[midway,right]{$1$}} ]]]\end{forest}}} +\dots + 
	\vcenter{\hbox{\begin{forest} [$s$ [$x_d$ [$t$, edge label={node[midway,right]{$1$}} ]]]\end{forest}}}
	\end{equation}
Observe that the above is an equation in a context of $\MP(\bell,\ZZ_2)$ (in fact, all the measurements appearing in it, namely the $x_i$ as well as those occurring in $s$ and $t$, necessarily form a context of $\bell$), so we can apply \Cref{prop:allMPequations}.
For a global assignment $\gs\colon\coprod_{i\in\II}\ZZ_2\to\ZZ_2$, both sides of \cref{eq:separatingsumsforbranches} have, under $\gs$, the same output as the protocol
\[
\vcenter{\hbox{\begin{forest} [$s$ [$x_1$ [$t$, edge label={node[midway,right]{$r+\gs(x_2)+\cdots+\gs(x_d)$}} ]]]\end{forest}}}
\] 
Thus \cref{eq:separatingsumsforbranches} holds for all global assignments, and hence is satisfied by the $\ZZ_2$-linear theory of $\MP(e,\ZZ_2)$ by \Cref{prop:allMPequations}.

The claim then follows by expressing the original protocol as a sum of branches using \Cref{lem:sumofbranches}, applying \cref{eq:separatingsumsforbranches} to each branch, and then recollecting the terms. 
\end{proof}

We pause to address a subtlety in the previous proof.
Namely, by assumption, the left-hand sides of these equations denote well-defined protocols.
However, the same is not obvious for terms appearing on the right-hand side.
The issue has to do with repeated measurements:
while the contexts appearing in a measurement protocol are distinct,
they may share variables (allowing such overlaps is in fact  crucial for the overall proof).
When applying \Cref{lem:separatingsums} to expand two different higher-arity branching levels of the protocol, 
it could happen that a variable $x_i$ occurs as a unary node twice in the same protocol in the end result, which \Cref{def:moregeneralmps} does not permit. 
Rather than introducing yet another generalisation of measurement protocols (``measurement protocols with potentially repeating nodes''), we treat each such expression unambiguously as shorthand for an honest protocol in the sense of \Cref{def:moregeneralmps}, as we now explain.

In a nutshell, the convention for branches is: given repeated occurrences of the same variable, we keep only the first one, provided they agree on the required outcome value, and otherwise the protocol equals zero.
Slightly more formally, we assert the following equations as definitions: 
for a measurement $x$, branches $s_1,s_2,s_3$, and an outcome $r \in \ZZ_2$,
\[\vcenter{\hbox{\begin{forest} [$s_1$ [$x$ [$s_2$, edge label={node[midway,right]{$r$}} [$x$ [$s_3$, edge label={node[midway,right]{$r$}} ]]]]]\end{forest}}} \defeq \vcenter{\hbox{\begin{forest} [$s_1$ [$x$ [$s_2$, edge label={node[midway,right]{$r$}}  [$s_3$ ]]]]\end{forest}}} 
\qquad \text{and} \qquad 
\vcenter{\hbox{\begin{forest} [$s_1$ [$x$ [$s_2$, edge label={node[midway,right]{$r$}} [$x$ [$s_3$, edge label={node[midway,right]{$r + 1$}} ]]]]]\end{forest}}}\defeq 0 .
\] 

Informally, the conventions for protocols more complicated than branches follow from this: just decompose them into branches and apply this convention branch-wise.
Formally, this is not fully rigorous without defining these protocols with repeated measurements and proving a version of \Cref{lem:sumofbranches} decomposing them into branches.
Hence we extend our conventions to cover cases in which the repeated measurement $x$ occurs once or twice as a branching node. When the first occurrence is multiplicative and the second is branching, we assert the following equations as definitions:  
\[\vcenter{\hbox{\begin{forest} [$p_1$ [$x$ [$p_2$, edge label={node[midway,right]{$0$}} [$x$ [$p_3$, edge label={node[midway,left]{$0$}} ] [$v.p_3$, edge label={node[midway,right]{$1$}} ]]]]]\end{forest}}}\defeq 
\vcenter{\hbox{\begin{forest} [$p_1$ [$x$ [$p_2$, edge label={node[midway,right]{$0$}}  [$p_3$ ]]]]\end{forest}}}
\qquad \text{and} \qquad 
\vcenter{\hbox{\begin{forest} [$p_1$ [$x$ [$p_2$, edge label={node[midway,right]{$1$}} [$x$ [$p_3$, edge label={node[midway,left]{$0$}} ] [$v.p_3$, edge label={node[midway,right]{$1$}} ]]]]]\end{forest}}}\defeq 
\vcenter{\hbox{\begin{forest} [$p_1$ [$x$ [$p_2$, edge label={node[midway,right]{$1$}}  [$v.p_3$ ]]]]\end{forest}}}
\]
Analogous equations are asserted when the first occurrence is branching and the second is multiplicative. 
When both occurrences are branching, only the first branching node remains as is, while the vectors that act on the continuation upon obtaining outcome $1$ are added together (after mapping them to the correct space):
\[\vcenter{\hbox{\begin{forest} dense [$p_1$ [$x$ [$p_2$, edge label={node[midway,left]{$0$}} [$x$ [$p_3$, edge label={node[midway,left]{$0$}} ] [$v_2.p_3$, edge label={node[midway,right]{$1$}} ]]] [$\pi_{s(p_2)}v_1.p_2$, edge label={node[midway,right]{$1$}} [$x$ [$\pi_{s(p_3)+1}v_1.p_3$, edge label={node[midway,left]{$0$}} ] [$(\pi_{s(p_3)+1}(v_1)+v_2).p_3$, edge label={node[midway,right]{$1$}} ]]]]]\end{forest}}}\defeq
\vcenter{\hbox{\begin{forest}  [$p_1$ [$x$ [$p_2$, edge label={node[midway,left]{$0$}} [$p_3$ ]] [$\pi_{s(p_2)}v_1.p_2$, edge label={node[midway,right]{$1$}} [$(\pi_{s(p_3)+1}(v_1)+v_2).p_3$ ]]]]\end{forest}}}\]
In the above, we think of the left-hand side as arising from an honest measurement protocol via \Cref{lem:separatingsums},
where $v_1$ (resp.~$v_2$) is the vector affecting the continuation of the first (resp.~second) occurrence of $x$.
The admissibility requirement in \Cref{def:moregeneralmps} ensures that the component of $v_1$ at $s(x)$  is zero.

What if, after repeated application of the equations above, we end up with more than two occurrences of a given variable in a single protocol?
In some sense, the issue does not arise: if we apply \Cref{lem:separatingsums} to an honest protocol, we end up with at most two occurrences of a given variable. Moreover, since the definitions above always prioritise the first occurrence, the order in which repeated pairs are handled does not matter (and in particular, neither does the order in which higher-arity nodes are eliminated via \Cref{lem:separatingsums}): each variable ultimately resolves to its earliest occurrence,
much like the operation of projecting a pair to its first coordinate defines the associative operation of selecting the head of any non-empty list.

We apply \Cref{lem:separatingsums} in the context of \Cref{thm:sumstozero}.
The following two lemmas implement Step 3 of removing higher-arity branching nodes: the first eliminates a single higher-arity branching node, and \Cref{lem:sumofsmallersums} removes all such nodes by iteration.

\begin{lemma}
Let $p$ be a measurement protocol with all multiplicative nodes unary, and $Q\colon \ZZ_2^{l}\to \ZZ_2^{s(p)+1}$ be a linear map where $l\geq m(p)+2$.  If $p$ has a $d$-ary branching node with $d>1$,
then there exist measurement protocols  $p_1, \ldots, p_d$,
each with $n(p_i)\leq n(p)$,
strictly fewer higher-arity branching levels than $p$,
and all multiplicative nodes unary, and linear maps $Q_i\colon \ZZ_2^{l_i}\to \ZZ_2^{s(p_i)+1}$ with $l_i\geq m(p_i)+2$
for $i=1,\dots, d$ such that 
\begin{equation}\label{eq:sumofsmallersums}\sum_{w\in \ZZ_2^{l}} Q(w).p\;=\,\sum_{i=1}^d\sum_{w\in \ZZ_2^{l_i}} Q_i(w).p_i .\end{equation}
\end{lemma}
\begin{proof}
The key idea is to apply \Cref{lem:separatingsums} across the action of $Q$. 
Fixing a $d$-ary level of $p$ that measures $x_1+\dots +x_d$, we can write $p$ as \[ \vcenter{\hbox{\begin{forest} [$q$ [$x_1+\dots +x_d$ [$q'$, edge label={node[midway,left]{$0$}} ] [$v.q'$, edge label={node[midway,right]{$1$}} ]]]\end{forest}}}\] 
and use \Cref{lem:separatingsums} to deduce
\[
\vcenter{\hbox{\begin{forest} [$q$ [$x_1+\dots +x_d$ [$q'$, edge label={node[midway,left]{$0$}} ] [$v.q'$, edge label={node[midway,right]{$1$}} ]]]\end{forest}}}
	= \vcenter{\hbox{\begin{forest} [$q$ [$x_1$ [$q'$, edge label={node[midway,left]{$0$}} ] [$v.q'$, edge label={node[midway,right]{$1$}} ] ]]\end{forest}}} + 
    \sum_{i=2}^d \left (
	\vcenter{\hbox{\begin{forest} [$q$ [$x_i$ [$q'$, edge label={node[midway,right]{$1$}} ]]]\end{forest}}} +
	\vcenter{\hbox{\begin{forest} [$q$ [$x_i$ [$v.q'$, edge label={node[midway,right]{$1$}} ]]]\end{forest}}} \right ).
\]

We then define $p_1$ and $p_i$ for $i=2,\dots, d$ by setting
\[p_1\defeq \vcenter{\hbox{\begin{forest} [$q$ [$x_1$ [$q'$, edge label={node[midway,left]{$0$}} ] [$v.q'$, edge label={node[midway,right]{$1$}} ] ]]\end{forest}}} 
\qquad \text{and} \qquad 
p_i\defeq\vcenter{\hbox{\begin{forest} [$q$ [$x_i$ [$q'$, edge label={node[midway,right]{$1$}} ]]]\end{forest}}} .\]
We set $l_1\defeq l$ and $l_i \defeq l+1$ for $i=2,\dots, d$. By construction,  each $p_i$ satisfies $n(p_i)\leq n(p)$, has
strictly fewer higher-arity branching levels than $p$
and all multiplicative nodes unary,  and $l_i\geq m(p_i)+2$.

It remains to define the linear maps $Q_i$ and verify \cref{eq:sumofsmallersums}.
We set $Q_1\defeq\pi_{s(p_1)+1}\circ Q$ and for $i=2,\dots,d$ we set
\[Q_i\defeq\begin{bmatrix} \pi_{s(p_i)+1}\circ Q & \iota_{s(p_i)+1}(v)\end{bmatrix}.\]

With these definitions, we now have that 
\[p=p_1 +\sum_{i=2}^d (p_i+\iota_{s(p_i)+1}(v).p_i)=p_1+\sum_{i=2}^d (p_i + Q_i(e_{l+1}).p_i)  \]
where $e_{l+1}=(0,\dots, 0,1)\in\ZZ_2^{l+1}$.
This decomposition does not depend on the initial choices of measurement settings, and hence goes through verbatim for $Q(w).p$.
This then gives 
\[Q(w).p=Q_1(w).p_1+\sum_{i=2}^d \left( Q_i(\iota_{l+1}(w)).p_i+Q_i(\iota_{l+1}(w)+e_{l+1}).p_i \right).\]
As $\ZZ_2^{l+1}= \im (\iota_{l+1})\oplus \ZZ_2 e_{l+1}$, summing over all $w\in \ZZ_2^l$ gives~\eqref{eq:sumofsmallersums}, concluding the proof.
\end{proof}

Repeating the previous Lemma recursively until all branching levels are unary yields the following.

\begin{lemma}\label{lem:sumofsmallersums}
Let $p$ be a measurement protocol with all multiplicative nodes unary, and $Q\colon \ZZ_2^{l}\to \ZZ_2^{s(p)+1}$ be a linear map where $l\geq m(p)+2$.
Then there are measurement protocols $p_i$ with all nodes unary and $n(p_i)\leq n(p)$, and linear maps $Q_i\colon \ZZ_2^{l_i}\to \ZZ_2^{s(p_i)+1}$ with $l_i\geq m(p_i)+2$
for $i=1,\dots, d$ such that
\begin{equation*}\sum_{w\in \ZZ_2^{l}} Q(w).p \;=\, \sum_{i=1}^d\sum_{w\in \ZZ_2^{l_i}} Q_i(w).p_i .\end{equation*}
\end{lemma}

\section{(Cohomo)logical consequences}
\label{sec:cohomology}

In this section we draw further consequences from our main result.
In particular, we use it to answer the question raised by Robert Raussendorf~\cite{raussendorf:cohomologicalframework,raussendorf:paradoxestowork} on how to lift  the cohomological criterion for contextuality to the adaptive case.
We do so in two distinct cohomological frameworks for contextuality:
the sheaf-theoretic (\v{C}ech) one of~\cite{abramskyetal:cohomology-of-contextuality,abramskyetal:contextualitycohomologyparadox} and the group-cohomological one of~\cite{raussendorf:cohomologicalframework,okay2017topological,raussendorf:roleofcohomology}, in which Raussendorf originally posed his question.
We treat each in turn,
and then also examine the flattening construction through the lens of Raussendorf's de-iffification~\cite{raussendorf:paradoxestowork} of conditional AvN arguments.
While the full details would take us outside the scope of this article, we sketch the main ideas here,
which we intend to develop further in a follow-up paper.

\subsection{\v{C}ech-cohomological witnesses}
\label{ssec:cechcohomology}
In the sheaf-theoretic framework of~\cite{abramskyetal:cohomology-of-contextuality,abramskyetal:contextualitycohomologyparadox,caru2017onthecohomology,caru2018towards}, strong contextuality is witnessed by a non-vanishing \v{C}ech cohomology class.
The following result states that an AvN argument over measurement protocols implies the existence of such a witness not just at the level of measurement protocols, but already over the original scenario.

To state the theorem, we need to introduce some notation not used elsewhere in this paper: if $S$ is a scenario, we write $\MP(S)$ for the set of general measurement protocols over $S$ as defined in~\cite{abramskyetal:comonadicview}, and $\MP_{\ZZ_2}(S)$ for the scenario of $\ZZ_2$-coarse-grained measurement protocols over $S$.
Finally, for an empirical model $e$ the notation $\CSC_{\ZZ_2}(e)$ means that $e$ is cohomologically strongly contextual with coefficients in $\ZZ_2$, as defined in~\cite{abramskyetal:cohomology-of-contextuality,abramskyetal:contextualitycohomologyparadox} -- that is, the strong contextuality of $e$ is witnessed by a non-vanishing \v{C}ech cohomology class.

\begin{theorem}\label{thm:adaptiveAvNsandcohomology}
Let $S=(X, \ctx, \ZZ_2)$ be a scenario and  $e: S$ an empirical model.
Then $\AvN_{\ZZ_2}(\MP_{\ZZ_2}(e))\Rightarrow \CSC_{\ZZ_2}(e)$.
That is, if the induced empirical model $\MP_{\ZZ_2}(e)$ on the scenario of $\ZZ_2$-coarse-grained measurement protocols exhibits AvN contextuality, then the original model $e$ is cohomologically strongly contextual.
\end{theorem}
\begin{proof}[(Proof sketch)] We prove the chain of implications
\[\AvN_{\ZZ_2}(\MP_{\ZZ_2}(e)) \Rightarrow\CSC_{\ZZ_2}(\MP_{\ZZ_2}(e))\Rightarrow   \CSC_{\ZZ_2}(\MP(e)) \Rightarrow \CSC_{\ZZ_2}(e)\]
The first implication is~\cite[Theorem 21]{abramskyetal:contextualitycohomologyparadox}. 
The second follows from the  existence of a canonical map of scenarios $\MP(S)\to \MP_{\ZZ_2}(S)$ (in the sense of~\cite{abramskyetal:comonadicview}),
together with the fact that cohomology cooperates suitably with maps of scenarios.

Finally, the vanishing of the cohomological obstruction to extending a local section of $e$ corresponds to a certain cochain being a coboundary in the relative cohomology.
Concretely, this amounts to a compatible family of $\ZZ_2$-linear combinations of local sections of $e$ extending the given local section, where compatibility is an analogue of no-signalling \cite{abramskyetal:cohomology-of-contextuality,abramsky2017contextuality}.
Therefore, just as an empirical model $e: S$ extends canonically to an empirical model $\MP(e): \MP(S)$, such a family extends to corresponding data for $\MP(e)$,
witnessing the vanishing of the cohomological obstruction over $\MP(e)$.
Taking contrapositives gives the final implication.
\end{proof}

The part of this argument requiring significant elaboration is the verification of the functoriality and naturality properties that ensure cohomology behaves well with respect to maps of scenarios.
While essentially straightforward, the details are somewhat lengthy.

By the theorem above, AvN contextuality of $\MP(e)$ implies the non-vanishing of the \v{C}ech-cohomological invariant on $e$.
Combining this with our main result (\Cref{thm:main}) yields an answer to Raussendorf's question within the sheaf-cohomological framework:
cohomological witnesses for contextuality extend to adaptive MBQC protocols. Moreover, in this framework, the witness lives already over the original scenario.

\subsection{Group-cohomological witnesses}
\label{ssec:groupcohomology}
Raussendorf originally formulated his question in a different cohomology framework~\cite{raussendorf:cohomologicalframework,okay2017topological,raussendorf:roleofcohomology} from the one considered above, based on group cohomology.
Those group-cohomological witnesses were formulated for sets of Pauli/Weyl operators closed under commuting products.
Aasn{\ae}ss~\cite{aasnaess2020cohomology,aasnaess2021dphil} made a direct comparison between the two frameworks,
first generalising the group-cohomological approach to a more abstract setting 
and then showing that a non-trivial group-cohomological obstruction implies a non-trivial \v{C}ech obstruction.
Note that this implication runs from group-cohomological to \v{C}ech-cohomological obstructions, so \Cref{thm:adaptiveAvNsandcohomology}, 
which establishes non-triviality of the \v{C}ech cohomology class, does not directly translate to group-cohomological witnesses;
we can nonetheless derive such witnesses for $\MP(e)$ from our main result, as we now describe.

The abstract formulation in \cite{aasnaess2021dphil} is in terms of $G$-bundle scenarios which are scenarios whose measurement set carry algebraic structure: each context is a commutative monoid on which the group $G$ of outcomes acts compatibly. We briefly describe how the measurement scenarios we consider can be interpreted as such, and draw the corresponding consequences from our main result.
The key idea is that AvN arguments are formulated for scenarios with algebraic structure on the \emph{outcomes};
the move is to lift that structure to the \emph{measurements} themselves.
Recall the setting of \Cref{ssec:AvNcontextuality}: a measurement scenario $S=(X, \ctx, O)$
where each measurement has outcomes valued in the abelian group $O = \ZZ_2$.
\footnote{\label{fn:module2}
More generally, as per \cref{fn:module}, one may take the outcome set $O$ to be an arbitrary abelian group, or even any $R$-module over a ring $R$, yielding analogous constructions; see \Cref{remark:moregeneral_module}.}
The abelian group structure of $\ZZ_2$ can be lifted from outcomes to measurements, giving rise to a partial algebraic structure.
Operationally, given two compatible $\ZZ_2$-valued measurements $x$ and $y$,
one may form a new $\ZZ_2$-valued measurement $x+y$ by jointly measuring $x$ and $y$ and adding their outcomes.
We have in fact been using this operation implicitly throughout, since it underlies the higher-arity nodes that arise in the condensing construction of \Cref{ssec:condensing}. One may also adjoin, for each $a \in \ZZ_2$, a constant measurement $c_a$ that always returns outcome $a$. The measurement $x + y$ is itself compatible with both $x$ and $y$, while each $c_a$ is compatible with every measurement,
so any context may be closed under these operations (binary addition and the two constants).

Under such closure and identification of operationally indistinguishable measurement procedures,
each context gives rise to a \emph{pointed Boolean group} $G(C)$:
an abelian group in which every element has order~$2$, with neutral element $c_0$, with an additional distinguished element $c_1$.
Equivalently, a Boolean group can be thought of as a $\ZZ_2$-vector space and the distinguished element singles out a vector, corresponding to moving to an affine space.
\footnote{\label{fn:linear-pBA}From a logical perspective, this structure captures the linear fragment of a Boolean algebra:
the group addition is exclusive-or, the constants $c_0$ and $c_1$ are the constants false and true, with negation being recovered as $m \mapsto m + c_1$; the multiplicative connectives (conjunction and disjunction) are absent.}
This same structure admits an equivalent description via a $\ZZ_2$-action $\theta \colon \ZZ_2 \times G(C) \to G(C)$,
required to be a monoid homomorphism (as per~\cite[Definitions 3.1.4 and 3.3.1]{aasnaess2021dphil}), or equivalently, to satisfy
$\theta(a,\, m + m') \;=\; m+\theta(a,\, m')$
for all $a \in \ZZ_2$ and $m, m' \in G(C)$ (as per~\cite{abramsky2026maifest}).
Note that this condition forces the action to be translation by constants (which form an image of $\ZZ_2$ inside $G(C)$):
setting $m' = c_0$ gives $\theta(a, m) = m + \theta(a, c_0)$ for all $m$, so $\theta$ is entirely determined by constants $c_a \defeq \theta(a, c_0)$.
Operationally, the non-trivial element $1 \in \ZZ_2$ acts as $m \mapsto m + c_1$, flipping the outcome
(that is, the negation of the logical interpretation in \cref{fn:linear-pBA}).

These pointed Boolean groups $G(C)$ on each context patch together into a \emph{partial pointed Boolean group} $G(S)$,
a partial algebra in the same spirit as partial Boolean algebras~\cite{kochen1967problem,abramsky2021logic}.
Namely, a partial pointed Boolean group is a set $G$ equipped with a reflexive symmetric (\emph{commeasurability}) relation $\odot$,
a partial binary (\emph{addition}) operation $+\colon{\odot}\to G$, and constants $c_0, c_1\in G$,
such that every set of pairwise commeasurable elements extends to a set of pairwise commeasurable elements forming a (total) pointed Boolean group under restriction of the operations.
\footnote{Concretely, this means that the operations (nullary $c_0, c_1$ and binary $+$) respect commeasurability --
$c_0\odot m$ and $c_1\odot m$ for every $m$, and $m_1+m_2\odot m'$ whenever $m_1,m_2,m'$ are pairwise commeasurable --
and that the axioms of pointed Boolean groups hold among any pairwise commeasurable elements. 

The partial pointed Boolean group $G(S)$ associated to a measurement scenario $S=(X,\ctx,O)$ is the free partial pointed Boolean group on the reflexive graph of measurement compatibility: generated by $\iota(x)$ for $x\in X$, subject to $\iota(x)\odot\iota(y)$ whenever $\{x,y\} \in \ctx$.
The explicit inductive construction of the free partial Boolean algebra on a reflexive graph given in \cite{abramsky2021logic,abramsky2022tacl} adapts directly to this setting, giving the free partial pointed Boolean group.}
The $\ZZ_2$-action extends compatibly to this partial algebraic structure, making $G(S)$ a $\ZZ_2$-bundle scenario in the sense of~\cite[Chapter 3]{aasnaess2021dphil}.

This structure $G(S)$ is the free partial pointed Boolean group generated by the scenario $S$.
The linear theory of an empirical model $e \colon S$ can now be incorporated into this structure. Recall that a linear equation over a context $C$ has the form $\sum_{x \in D} x = a$ where $D \subseteq C$ and $a \in \ZZ_2$ (see \Cref{ssec:AvNcontextuality}). Since all $\iota(x)$ for $x \in D$ are compatible in $G(S)$ (because $D$ is a context of $S$),
this equation makes sense as an equation between elements of $G(S)$: $\sum_{x \in D} \iota(x) = c_a$.
Quotienting $G(S)$ by the linear theory of $e$ thus yields a further partial pointed Boolean group $G(S,e) \defeq G(S)/\Th(e)$.

A global assignment $\gs \colon X \to \ZZ_2$ extends uniquely to a homomorphism $G(S) \to \ZZ_2$ of partial pointed Boolean groups by freeness,
and it satisfies the linear theory of $e$ precisely when this homomorphism factors through $G(S,e)$.
Therefore, the model $e$ is $\AvN_{\ZZ_2}$-contextual, \ie its $\ZZ_2$-linear theory is inconsistent,
if and only if no such homomorphism $\tilde{\gs} \colon G(S,e) \to \ZZ_2$ exists.
This is equivalent to saying that the sequence
\begin{equation}\label{eq:bundlesequence}
\begin{tikzcd}
  \ZZ_2 \arrow[r, "c", hook] &
  G(S,e) \arrow[r, "{[\cdot]_\theta}", two heads] &
  G(S,e)/\ZZ_2
\end{tikzcd}
\end{equation}
where $c$ is the inclusion $a \mapsto c_a$
and $[\cdot]_\theta$ is the quotient by the $\ZZ_2$-translation action,
has no left splitting, that is, no homomorphism $\tilde{\gs} \colon G(S,e) \to \ZZ_2$ of partial pointed Boolean groups
satisfying $\tilde\gs \circ c = \mathrm{id}_{\ZZ_2}$.
Since any homomorphism of partial pointed Boolean groups to $\ZZ_2$ sends $c_a \mapsto a$ automatically,
this is precisely a satisfying assignment for the linear theory of $e$.
Non-existence of such a splitting is exactly a cohomological obstruction in the sense of~\cite{aasnaess2021dphil}.
\footnote{This property is phrased as a splitting of the whole sequence~\eqref{eq:bundlesequence}, rather than as a retraction of $c$ alone:
a version of the splitting lemma~\cite[Lemma~3.1.1]{aasnaess2021dphil} establishes a bijection between such left splittings and \emph{trivialisations} of the bundle $[\cdot]_\theta\colon G(S,e)\twoheadrightarrow G(S,e)/\ZZ_2$.  It is the (non-)existence of a trivialisation of the bundle that the group-cohomological obstruction of~\cite{aasnaess2021dphil} measures.}

Applying this to $\MP_{\ZZ_2}(e)$: our main result (\Cref{thm:main}) produces an AvN argument over $\MP_{\ZZ_2}(e)$,
which by the above witnesses a non-trivial group-cohomological obstruction for $G(\MP_{\ZZ_2}(S),\, \MP_{\ZZ_2}(e))$.
In contrast to the \v{C}ech-cohomological case (\Cref{thm:adaptiveAvNsandcohomology}), this obstruction does \emph{not} propagate back to a group-cohomological obstruction for $G(S,e)$:
the reason is the same as discussed there, namely that our construction yields an AvN argument over measurement protocols, not necessarily over the original scenario.

\begin{remark}\label{remark:moregeneral_module}
The analysis above generalises beyond the $\ZZ_2$ setting considered here (see \cref{fn:module,fn:module2}).
For $O$ an arbitrary abelian group, the construction above goes through almost verbatim: one adjoins constants $c_a$ for each element $a\in O$, and closes each context under pairwise addition and constants, exactly as before.
This yields a partial abelian group $G(S)$ containing, via the homomorphism $a\mapsto c_a$, an image of $O$ commeasurable with every element of $G(S)$.
The $O$-action is again given by translation $m\mapsto m+c_a$.
This makes $G(S)$ an $O$-bundle scenario in the sense of~\cite[Chapter 3]{aasnaess2021dphil};
the analysis above carries through in this setting.

More generally, if $O$ is an $R$-module over a commutative ring $R$,
one lifts not only module addition and the constants $c_a$ for elements of $O$, but additionally closes each context under scalar multiplication by ring elements, equipping $G(S)$ with a total, unary operation $r\cdot \colon G(S)\to G(S)$ for each $r \in R$.
$G(S)$ then becomes a partial $R$-module containing an image of $O$ as a distinguished $R$-linear submodule, commeasurable with every element.
A similar analysis can be carried through in this setting.
\end{remark}

\subsection{De-iffifying adaptive AvN arguments}
\label{ssec:deiffify}

Our flattening construction also admits a logical perspective that sharpens the picture above.
Adaptive protocols allow for systems of \emph{conditional} equations providing AvN arguments, where the conditioning on variables corresponds to the adaptivity.
A concrete example of such a system of conditional equations appears in \cite{abramskyetal:minimumquantumresources}.
Such arguments have been dubbed ``iffy proofs'' by Raussendorf \cite{raussendorf:paradoxestowork}.
Thus we can interpret our flattening construction as a process of ``de-iffification''.

The situation is, however, somewhat more subtle:
arbitrary conditional systems can capture arbitrary supports,  hence in that setting inconsistency coincides with strong contextuality.
However, \Cref{thm:adaptiveAvNsandcohomology} shows that all-versus-nothing contextuality over measurement protocols is at most as strong as \v{C}ech-cohomological contextuality.
This leads to a further consequence:
$\AvN_{\ZZ_2}(e)$ is \emph{not} implied by $\AvN_{\ZZ_2}(\MP_{\ZZ_2}(e))$, and therefore not by $\CSC_{\ZZ_2}(e)$ either. In particular, the adaptive AvN arguments produced by our construction are genuinely new:  they go beyond what non-adaptive AvN arguments on the underlying scenario can detect.
This is witnessed by the iffy proof in~\cite[Section 5]{abramskyetal:minimumquantumresources}, since one can verify that the ordinary $\ZZ_2$-linear theory of the underlying model in  that example is consistent. 

By \Cref{ssec:groupcohomology}, the same AvN argument also witnesses a non-trivial group-cohomological obstruction for $G(\MP_{\ZZ_2}(S),\, \MP_{\ZZ_2}(e))$.
Unlike the \v{C}ech case, this cannot in general propagate back to a group-cohomological obstruction for $G(S,e)$ itself:
such an obstruction would amount to $\AvN_{\ZZ_2}(e)$, which the iffy proof example just discussed already shows is too strong.

\section{Further directions}
\label{sec:furtherdirections}

We close by outlining some directions for further work suggested by our results:
\begin{itemize}

\item Does the converse to \Cref{thm:adaptiveAvNsandcohomology} hold?
That is, does (\v{C}ech-)cohomological strong contextuality of $e$ imply that $\MP_{\ZZ_2}(e)$ exhibits algebraic contextuality?
Equivalently, does every cohomological proof of strong contextuality  give rise to an AvN argument over measurement protocols?

\item Our construction produces AvN arguments for the scenario $\MP(\bell,\ZZ_2)$ of measurement protocols.
What do these translate to when expressed in terms of the original measurements of $\bell$?
More precisely, which conditional AvN arguments (`iffy proofs', in the sense of~\cite{raussendorf:paradoxestowork,abramskyetal:minimumquantumresources}) on $\bell$ correspond to the AvN arguments our construction produces?

\item In~\cite{abramsky2017contextual}, a quantitative refinement of Raussendorf's result was established, with an inequality relating the success probability of an MBQC, the nonlinearity of the computed Boolean function, and the contextual fraction (a measure of contextuality) of the underlying empirical model.
Can an analogous quantitative refinement of our main result be developed for adaptive protocols?

\item Here we focussed on dichotomic measurements with outcomes in $\ZZ_2$, typically quantum-realised using qubits.
A natural extension considers measurements with outcomes in $\ZZ_d$ for $d > 2$, corresponding to qudit-based MBQC.
Contextuality and AvN arguments have been studied for that setting in~\cite{frembs2018contextuality,frembs2026no,GogiosoZeng:AvNs}.
Do analogous results hold for adaptive $\ZZ_d$-linear protocols? 

\item In \cite{abramsky2017complete}, there is a complete classification of all quantum AvN arguments for stabiliser states.
Does an analogous classification exist for adaptive AvN arguments of the kind produced by our construction?

\item We have focussed on algebraic AvN paradoxes.
How do these relate to the logical paradoxes (contradictions verified in partial Boolean algebras) studied by Kochen and Specker in their seminal work on contextuality  \cite{kochen1967problem,abramsky2021logic}?
In particular, how can one capture adaptive measurements on partial Boolean algebras? Can an analogue of our flattening construction be developed in that setting?

\item The scenario $\MP(\bell,\ZZ_2)$ is constructed from $\bell$. Can this be phrased for arbitrary algebraic contextuality scenarios, rather than just those of Bell type? Is this construction functorial with respect to morphisms of scenarios, which give rise to simulations between empirical models~\cite{karvonen2018categories,barbosa2023closing}, and does our main result interact naturally with such morphisms?

\item Building on the previous point: in~\cite{abramskyetal:comonadicview} a more general notion of measurement protocol was structured as a comonad on a category of empirical models. Can the version of $\MP$ studied here  be given a similar comonadic structure, possibly graded by the depth or other parameters of the protocol?

\item Throughout this paper we have worked within the MBQC paradigm.
Can our results be related to the circuit model of quantum computation, perhaps via the notion of sequential contextuality \cite{mansfield2018quantum}?

\item 
Our construction yields explicit AvN witnesses for adaptive protocols.
Can these be used for self-testing~\cite{mayers2004selftesting,supic2020selftesting} of quantum resources,
certifying not only the contextuality but the specific quantum state and measurements?
There is reason to expect so: AvN arguments, and the closely related linear constraint system games~\cite{cleve2014characterization} (which are in a sense a state-independent analogue of AvN arguments),
have already enabled self-testing in the non-adaptive setting~\cite{kaniewski2016analytic,coladangelo2017robust},
including scalable self-testing of graph states~\cite{baccari2020scalable} and CSS codewords~\cite{hart2025manybody}. Whether adaptive AvN witnesses admit analogous results, where the feed-forward structure itself must be certified, is an open problem.

\item There is a connection between contextuality and quantum error correction: \cite{khu2026contextualityQEC} show that, given any subsystem stabiliser code with two or more gauge qubits, the closure of its measurements as a partial abelian group yields a strong contextuality argument. Such contextuality, in the spirit of Kirby and Love~\cite{kirby2019contextuality}, is algebraic in nature, albeit state-independent. They further show that code-switching schemes for universal fault tolerance -- such as the doubled colour codes of Bravyi and Cross~\cite{bravyi2015doublecolorcodes} -- inherit such contextuality. However, their argument is \emph{static}, concerning the fixed algebraic structure of the code-switching gadget as a whole, rather than the protocol's adaptive unfolding. Could the techniques developed in this work deepen this connection by witnessing contextuality directly in the code-switching protocol itself?

\item Contextuality has been used for quantum hardware benchmarking: in~\cite{kumar2025quantumclassical}, AvN-type constraints are exploited for error detection beyond standard randomised benchmarking.
Our result provides explicit AvN witnesses for any adaptive $\ZZ_2$-linear MBQC computing a non-affine function.
Can such witnesses be used to certify or benchmark the correct functioning of adaptive quantum hardware, where one must account for the feed-forward structure?
\end{itemize}

\addtocontents{toc}{\SkipTocEntry}
\section*{Acknowledgements}
    This work was supported by
    the Horizon Europe project FoQaCiA, \textit{Foundations of Quantum Computational Advantage}, GA no.~\href{https://doi.org/10.3030/101070558}{101070558} (UCL and INL);
    the Engineering and Physical Sciences Research Council (EPSRC, U.K.) fellowship \textit{Resources in Computation}, \href{https://gtr.ukri.org/projects?ref=EP%2FV040944%2F1}{EP/V040944/1} (UCL);
    and the Fundação para a Ciência e Tecnologia (FCT, Portugal) project KaleidosQope, \textit{Contextual partial views: a logical foundation for quantum computational advantage}, \href{https://www.doi.org/10.54499/2023.13603.PEX}{2023.13603.PEX} (INL).

\bibliographystyle{plainurlarxiv}
\bibliography{refs}

@inproceedings{aasnaess2020cohomology,
  title     = {Cohomology and the algebraic structure of contextuality in measurement based quantum computation},
  author    = {Aasn{\ae}ss, Sivert},
  booktitle = {16th International Conference on Quantum Physics and Logic (QPL 2019)},
  editor    = {Coecke, Bob and Leifer, Matthew},
  series    = {Electronic Proceedings in Theoretical Computer Science},
  volume    = {318},
  pages     = {242--253},
  doi       = {10.4204/EPTCS.318.15},
  year      = {2020},
  month     = May,
  publisher = {Open Publishing Association}, 
}

@phdthesis{aasnaess2021dphil,
    author = {Aasn{\ae}ss, Sivert},
    title  = {Comparing two cohomological obstructions for contextuality, and a generalised construction of quantum advantage with shallow circuits},
    year   = {2021},
    school = {University of Oxford},
    type   = {{DPhil} thesis},
    url    = {https://ora.ox.ac.uk/objects/uuid:a31ce72c-3d5e-4ca2-afa3-a38afbb93833/},
    eprint    = {2212.09382},
    archivePrefix = {arXiv},
    primaryClass  = {quant-ph},
}

@article{ab,
  title={The sheaf-theoretic structure of non-locality and contextuality},
  author={Abramsky, Samson and Brandenburger, Adam},
  journal={New Journal of Physics},
  volume={13},
  number={11},
  pages={113036},
  year={2011},
  month = Nov,
  publisher={IOP Publishing},
  doi = {10.1088/1367-2630/13/11/113036},
}

@inproceedings{abramskyetal:cohomology-of-contextuality,
  title={The cohomology of non-locality and contextuality},
  author={Abramsky, Samson and Mansfield, Shane and Barbosa, Rui Soares},
  booktitle = {8th International Workshop on Quantum Physics and Logic (QPL 2011)},
  editor    = {Jacobs, Bart and Selinger, Peter and Spitters, Bas},
  series = {Electronic Proceedings in Theoretical Computer Science},
  volume    = {95},
  pages={1--14},
  year={2012},
  month = Oct,
  doi={10.4204/EPTCS.95.1},
  publisher = {Open Publishing Association}, 
}

@inproceedings{abramskyetal:contextualitycohomologyparadox,
  title = {Contextuality, cohomology and paradox},
  author =  {Abramsky, Samson and Barbosa, Rui Soares and Kishida, Kohei and Lal, Raymond and Mansfield, Shane},
  booktitle = {24th EACSL Annual Conference on Computer Science Logic (CSL 2015)},
  editor =  {Kreutzer, Stephan},
  series =  {Leibniz International Proceedings in Informatics (LIPIcs)},
  volume =  {41},
  pages = {211--228},
  year =  {2015},
  month = Sep,
  doi =   {10.4230/LIPIcs.CSL.2015.211},
  publisher = {Schloss Dagstuhl -- Leibniz-Zentrum f{\"u}r Informatik},
}

@article{abramsky2017contextual,
  title={Contextual fraction as a measure of contextuality},
  author={Abramsky, Samson and Barbosa, Rui Soares and Mansfield, Shane},
  journal={Physical Review Letters},
  volume={119},
  number={5},
  pages={050504},
  year={2017},
  month = Aug,
  publisher = {American Physical Society},
  doi={10.1103/PhysRevLett.119.050504},
}

@article{abramsky2017complete,
  title={A complete characterization of all-versus-nothing arguments for stabilizer states},
  author={Abramsky, Samson and Barbosa, Rui Soares and Car{\`u}, Giovanni and Perdrix, Simon},
  journal={Philosophical Transactions of the Royal Society A: Mathematical, Physical and Engineering Sciences},
  volume={375},
  number={2106},
  pages={20160385},
  year={2017},
  month = Nov,
  publisher={The Royal Society Publishing},
  doi = {10.1098/rsta.2016.0385},
}

@incollection{abramsky2017contextuality,
  title={Contextuality: At the borders of paradox},
  author={Abramsky, Samson},
  booktitle={Categories for the Working Philosopher},
  editor={Landry, Elaine},
  isbn = {9780198748991},
  chapter = {11},
  pages={262--285},
  year={2017},
  month = Nov,
  publisher={Oxford University Press},
  doi = {10.1093/oso/9780198748991.003.0011},
}

@inproceedings{abramskyetal:minimumquantumresources,
  title = {Minimum quantum resources for strong non-locality},
  author = {Abramsky, Samson and Barbosa, Rui Soares and Car{\`u}, Giovanni and de Silva, Nadish and Kishida, Kohei and Mansfield, Shane},
  booktitle =	{12th Conference on the Theory of Quantum Computation, Communication and Cryptography (TQC 2017)},
  editor = {Wilde, Mark M},
  series = {Leibniz International Proceedings in Informatics (LIPIcs)}, 
  volume = {73},
  pages =	{9:1--9:20},
  year = {2018},
  month = Mar,
  doi = {10.4230/LIPICS.TQC.2017.9},
  publisher = {Schloss Dagstuhl -- Leibniz-Zentrum f{\"u}r Informatik},
}

@inproceedings{abramskyetal:comonadicview,
  title={A comonadic view of simulation and quantum resources},
  author={Abramsky, Samson and Barbosa, Rui Soares and Karvonen, Martti and Mansfield, Shane},
  booktitle={34th Annual ACM/IEEE Symposium on Logic in Computer Science (LICS 2019)},
  articleno = {23},
  pages={23:1--23:12},
  year={2019},
  month = Aug,
  doi={10.1109/LICS.2019.8785677}, 
  publisher={IEEE}
}

@inproceedings{abramsky2021logic,
  title = {The logic of contextuality},
  author={Abramsky, Samson and Barbosa, Rui Soares},
  booktitle={29th EACSL Annual Conference on Computer Science Logic (CSL 2021)},
  editor = {Baier, Christel and Goubault-Larrecq, Jean},
  series = {Leibniz International Proceedings in Informatics (LIPIcs)}, 
  volume={183},
  pages={5:1--5:18},
  year={2021},
  month = Jan,
  doi =	{10.4230/LIPIcs.CSL.2021.5},
  publisher = {Schloss Dagstuhl -- Leibniz-Zentrum f{\"u}r Informatik},
}

@misc{abramsky2022tacl,
  author       = {Abramsky, Samson and Barbosa, Rui Soares},
  title        = {Contextuality in logical form: Duality for transitive partial {CABA}s},
  year         = {2022},
  month        = Jun,
  howpublished = {Talk presented at 10th Conference on Topology, Algebra, and Categories in Logic (TACL 2022)},
  url          = {https://www.mat.uc.pt/~tacl2022/pdfs/TACL_2022_paper_19.pdf},
  nonote       = {Abstract available at \url{https://www.mat.uc.pt/~tacl2022/pdfs/TACL_2022_paper_19.pdf}},
}

@incollection{abramsky2026maifest,
    title = {Contextuality, algebraic structure, and state-independence},
    author = {Abramsky, Samson},
    booktitle = {{M}ai {G}ehrke on Logic, Algebra and Duality},
    editor = {Reggio, Luca and van Gool, Sam and Fussner, Wesley},
    series = {Outstanding Contributions to Logic},
    novolume = {},
    nochapter = {},
    nopages = {},
    year = {2026},
    nomonth = 0,
    publisher = {Springer},
    note = {To appear},
    nodoi = {}
}

@article{anders2009computational,
  title = {Computational power of correlations},
  author = {Anders, Janet and Browne, Dan E.},
  journal = {Physical Review Letters},
  volume = {102},
  number = {5},
  pages = {050502},
  year = {2009},
  month = Feb,
  publisher = {American Physical Society},
  doi = {10.1103/PhysRevLett.102.050502},
}

@misc{arkhipov2012extending,
      title={Extending and characterizing quantum magic games},
      author = {Arkhipov, Alex},
      year={2012},
      month = Sep,
      eprint={1209.3819},
      archivePrefix={arXiv},
      primaryClass={quant-ph},
      url={https://arxiv.org/abs/1209.3819}, 
      doi = {10.48550/arXiv.1209.3819},
}

@article{baccari2020scalable,
  title = {Scalable {B}ell inequalities for qubit graph states and robust self-testing},
  author = {Baccari, Flavio and Augusiak, Remigiusz and {\ifmmode\check{S}\else\v{S}\fi{}}upi{\ifmmode\acute{c}\else\'{c}\fi{}}, Ivan and Tura, Jordi and Ac\'{\i}n, Antonio},
  journal = {Physical Review Letters},
  volume = {124},
  number = {2},
  pages = {020402},
  year = {2020},
  month = Jan,
  publisher = {American Physical Society},
  doi = {10.1103/PhysRevLett.124.020402},
}

@phdthesis{barbosa2015dphil,
    author = {Barbosa, Rui Soares},
    title = {Contextuality in quantum mechanics and beyond},
    year = {2015},
    school = {University of Oxford},
    type = {{DPhil} thesis},
}

@incollection{barbosa2023closing,
  title       = {Closing {B}ell: Boxing black box simulations in the resource theory of contextuality},
  author        = {Barbosa, Rui Soares and Karvonen, Martti and Mansfield, Shane},
  booktitle = {{S}amson {A}bramsky on Logic and Structure in Computer Science and Beyond},
  editor = {Palmigiano, Alessandra and Sadrzadeh, Mehrnoosh},
  series = {Outstanding Contributions to Logic},
  volume = {25},
  chapter = {13},
  pages = {475--529},
  year          = {2023},
  month = Aug,
  publisher = {Springer},
  doi = {10.1007/978-3-031-24117-8_13},
}

@article{bell1964einstein,
  title={On the {E}instein {P}odolsky {R}osen paradox},
  author={Bell, John S},
  journal={Physics Physique Fizika},
  volume={1},
  number={3},
  pages={195--200},
  year={1964},
  month = Nov,
  publisher = {American Physical Society},
  doi = {10.1103/PhysicsPhysiqueFizika.1.195},
}

@article{bell1966problem,
  author = {Bell, John S},
  title = {On the problem of hidden variables in quantum mechanics},
  journal = {Reviews of Modern Physics},
  volume = {38},
  number = {3},
  pages = {447--452},
  year = {1966},
  month = Jul,
  publisher = {American Physical Society},
  doi = {10.1103/RevModPhys.38.447}
}

@misc{bravyi2015doublecolorcodes,
      title={Double color codes},
      author={Bravyi, Sergey and Cross, Andrew},
      year={2015},
      month = Sep,
      eprint={1509.03239},
      archivePrefix={arXiv},
      primaryClass={quant-ph},
      url={https://arxiv.org/abs/1509.03239}, 
      doi = {10.48550/arXiv.1509.03239},
}

@article{bravyi2018quantum,
  title={Quantum advantage with shallow circuits},
  author={Bravyi, Sergey and Gosset, David and K{\"o}nig, Robert},
  journal={Science},
  volume={362},
  number={6412},
  pages={308--311},
  year={2018},
  month =Oct,
  publisher={American Association for the Advancement of Science},
  doi = {10.1126/science.aar3106},
}

@article{briegel2009mbqc,
    title = {Measurement-based quantum computation},
    author = {Briegel, Hans J. and Browne, Dan E. and D{\"u}r, Wolfgang and Raussendorf, Robert and van den Nest, Maarten},
    journal = {Nature Physics},
    volume = {5},
    number = {1},
    pages = {19--26},
    year = {2009},
    month = Jan,
    doi = {10.1038/nphys1157},
    publisher = {Nature Publishing Group}
}

@inproceedings{caru2017onthecohomology,
  title={On the cohomology of contextuality},
  author={Car{\`u}, Giovanni},
  booktitle={13th International Conference on Quantum Physics and Logic (QPL 2016)},
  editor    = {Duncan, Ross and Heunen, Chris},
  series    = {Electronic Proceedings in Theoretical Computer Science},
  volume = {236},
  pages     = {21--39},
  doi       = {10.4204/EPTCS.236.2},
  year={2017},
  month = Jan,
  publisher = {Open Publishing Association}, 
}

@misc{caru2018towards,
    title={Towards a complete cohomology invariant for non-locality and contextuality},
    author={Car{\`u}, Giovanni},
    year={2018},
    month= Jul,
    eprint={1807.04203},
    archivePrefix={arXiv},
    primaryClass={quant-ph},
    url={https://arxiv.org/abs/1807.04203},
    doi = {10.48550/arXiv.1807.04203}
}

@inproceedings{cleve2014characterization,
      title={Characterization of binary constraint system games}, 
      author={Cleve, Richard and Mittal, Rajat},
      booktitle = {41st International Colloquium on Automata, Languages, and Programming (ICALP 2014)},
      editor = {Esparza, Javier and Fraigniaud, Pierre and Husfeldt, Thore and Koutsoupias, Elias},
      series = {Lecture Notes in Computer Science},
      volume = {8572},
      pages = {320--331},
      year={2014},
      month = Jun,
      doi = {10.1007/978-3-662-43948-7_27},
      publisher = {Springer},
}

@misc{coladangelo2017robust,
      title={Robust self-testing for linear constraint system games}, 
      author={Coladangelo, Andrea and Stark, Jalex},
      year={2017},
      month = Sep,
      eprint={1709.09267},
      archivePrefix={arXiv},
      primaryClass={quant-ph},
      url={https://arxiv.org/abs/1709.09267}, 
      doi = {10.48550/arXiv.1709.09267},
}

@misc{frembs2026no,
  title={No quantum solutions to linear constraint systems from monomial measurement-based quantum computation in odd prime dimension},
  author={Frembs, Markus and Okay, Cihan and Chung, Ho Yiu},
  year={2026},
  month = Jan,
  eprint    = {2601.11367},
  archivePrefix = {arXiv},
  primaryClass  = {quant-ph},
      url={https://arxiv.org/abs/2601.11367},
      doi = {10.48550/arXiv.2601.11367},
}

@article{frembs2018contextuality,
  title = {Contextuality as a resource for measurement-based quantum computation beyond qubits},
  author = {Frembs, Markus and Roberts, Sam and Bartlett, Stephen D},
  journal = {New Journal of Physics},
  volume = {20},
  number = {10},
  pages = {103011},
  year = {2018},
  month = Oct,
  publisher = {IOP Publishing},
  doi = {10.1088/1367-2630/aae3ad},
}

@article{GogiosoZeng:AvNs,
  title = {Generalised {M}ermin-type non-locality arguments},
  author = {Gogioso, Stefano and Zeng, William},
  journal = {Logical Methods in Computer Science},
  volume = {15},
  number = {2},
  eid = {3},
  pages = {3:1--3:51},
  year = {2019},
  month = Apr, 
  doi = {10.23638/lmcs-15(2:3)2019},
}

@incollection{greenberger1989going,
  title={Going beyond {B}ell's theorem},
  author={Greenberger, Daniel M and Horne, Michael A and Zeilinger, Anton},
  booktitle={Bell's Theorem, Quantum Theory and Conceptions of the Universe},
  editor = {Kafatos, Menas},
  series = {Fundamental Theories of Physics},
  volume = {37},
  pages={69--72},
  year={1989},
  month = Oct,
  publisher={Springer},
  doi = {10.1007/978-94-017-0849-4_10},
}

@article{hardy1993nonlocality,
  title={Nonlocality for two particles without inequalities for almost all entangled states},
  author={Hardy, Lucien},
  journal={Physical Review Letters},
  volume={71},
  number={11},
  pages={1665--1668},
  year={1993},
  month = Sep,
  publisher = {American Physical Society},
  doi = {10.1103/PhysRevLett.71.1665},
}

@misc{hart2025manybody,
  title={Many-body contextuality and self-testing quantum matter via nonlocal games},
  author={Hart, Oliver and Stephen, David T and Wickenden, Evan and Nandkishore, Rahul},
  year = {2025},
  month = Dec,
  eprint    = {2512.16886},
  archivePrefix = {arXiv},
  primaryClass  = {quant-ph},
  url={https://arxiv.org/abs/2512.16886}, 
  doi = {10.48550/arXiv.2512.16886},
}

@article{howard2014contextuality,
  title={Contextuality supplies the {`}magic{'} for quantum computation},
  author={Howard, Mark and Wallman, Joel and Veitch, Victor and Emerson, Joseph},
  journal={Nature},
  volume={510},
  number={7505},
  pages={351--355},
  year={2014},
  month=Jun,
  publisher={Nature Publishing Group},
  doi = {10.1038/nature13460},
}

@article{kaniewski2016analytic,
  title = {Analytic and nearly optimal self-testing bounds for the {C}lauser-{H}orne-{S}himony-{H}olt and {M}ermin inequalities},
  author = {Kaniewski, J{\ifmmode\mbox{\k{e}}\else\k{e}\fi{}}drzej},
  journal = {Physical Review Letters},
  volume = {117},
  number = {7},
  pages = {070402},
  year = {2016},
  month=Aug,
  publisher = {American Physical Society},
  doi = {10.1103/PhysRevLett.117.070402},
}

@inproceedings{karvonen2018categories,
  title={Categories of empirical models},
  author={Karvonen, Martti},
  booktitle={15th International Conference on Quantum Physics and Logic (QPL 2018)},
  editor    = {Selinger, Peter and Chiribella, Giulio},
  series    = {Electronic Proceedings in Theoretical Computer Science},
  volume = {287},
  pages     = {239--252},
  doi       = {10.4204/EPTCS.287.14},
  year={2019},
  month = Jan,
  publisher = {Open Publishing Association}, 
}

@article{khu2026contextualityQEC,
  title = {Contextuality of quantum error-correcting codes},
  author = {Khu, Derek and Tanggara, Andrew and Jin, Chao and Bharti, Kishor},
  journal = {PRX Quantum},
  volume = {7},
  number = {1},
  pages = {010319},
  year = {2026},
  month = Jan,
  publisher = {American Physical Society},
  doi = {10.1103/7zb5-vs4x},
}

@article{kirby2019contextuality,
  title = {Contextuality test of the nonclassicality of variational quantum eigensolvers},
  author = {Kirby, William M. and Love, Peter J.},
  journal={Physical Review Letters},
  volume = {123},
  number = {20},
  pages = {200501},
  year = {2019},
  month = Nov,
  publisher = {American Physical Society},
  doi = {10.1103/PhysRevLett.123.200501},
}

@article{kirby2020classical,
  title={Classical simulation of noncontextual {P}auli {H}amiltonians},
  author={Kirby, William M and Love, Peter J},
  journal={Physical Review A},
  volume={102},
  number={3},
  pages={032418},
  year={2020},
  month = Sep,
  publisher = {American Physical Society},
  doi={10.1103/PhysRevA.102.032418},
}

@article{kochen1967problem,
  title={The problem of hidden variables in quantum mechanics},
  author={Kochen, Simon and Specker, Ernst P},
  journal={Journal of Mathematics and Mechanics},
  volume={17},
  number={1},
  pages={59--87},
  year={1967},
  month = Jul,
  doi={10.1512/iumj.1968.17.17004},
}

@misc{kumar2025quantumclassical,
  title={Quantum-classical separation in bounded-resource tasks arising from measurement contextuality},
  author = {{Google Quantum AI (Shashwat Kumar, Eliott Rosenberg,  Alejandro {Grajales Dau}, Rodrigo Corti{\~n}as, et al.)}},
  noauthor = {Kumar, Shashwat and Rosenberg, Eliott and {Grajales Dau}, Alejandro and Corti{\~n}as, Rodrigo and others},
  year={2025},
  month = Dec,
  eprint    = {2512.02284},
  archivePrefix = {arXiv},
  primaryClass  = {quant-ph},
  url = {https://arxiv.org/abs/2512.02284},
  doi = {10.48550/arXiv.2512.02284},
}

@article{lanyon2013mbqc,
  title = {Measurement-based quantum computation with trapped ions},
  author = {Lanyon, Benjamin P and Jurcevic, Petar and Zwerger, Michael and Hempel, Cornelius and Martinez, Esteban A and D{\"u}r, Wolfgang and Briegel, Hans J and Blatt, Rainer and Roos, Christian F},
  journal = {Physical Review Letters},
  volume = {111},
  number = {21},
  pages = {210501},
  year = {2013},
  month = Nov,
  publisher = {American Physical Society},
  doi = {10.1103/PhysRevLett.111.210501},
}

@article{larasatichoi2026circuitvsmbqc,
  title = {Circuit-based vs. measurement-based quantum computing: a comparative analysis, layered metrics, and decision flow for approach selection},
  author = {Larasati,  Harashta Tatimma and Choi,  Byung-Soo},
  journal = {EPJ Quantum Technology},
  volume = {13},
  number = {1},
  pages = {39},
  year = {2026},
  month = Feb, 
  publisher = {Springer Nature},
  doi = {10.1140/epjqt/s40507-026-00483-1},
}

@article{mansfield2018quantum,
  title={Quantum advantage from sequential-transformation contextuality},
  author={Mansfield, Shane and Kashefi, Elham},
  journal={Physical Review Letters},
  volume={121},
  number={23},
  pages={230401},
  year={2018},
  month = Dec,
  publisher = {American Physical Society},
  doi = {10.1103/PhysRevLett.121.230401},
}

@article{mayers2004selftesting,
  title={Self testing quantum apparatus},
  author={Mayers, Dominic and Yao, Andrew},
  journal={Quantum Information and Computation},
  volume={4},
  number={4},
  pages={273--286},
  year={2004},
  month = Jul,
  publisher = {Rinton Press},
  doi = {10.26421/QIC4.4-3},
  }

@article{mermin1990extreme,
  title={Extreme quantum entanglement in a superposition of macroscopically distinct states},
  author={Mermin, N David},
  journal={Physical Review Letters},
  volume={65},
  number={15},
  pages = {1838--1840},
  year={1990},
  month = Oct,
  publisher = {American Physical Society},
  doi = {10.1103/PhysRevLett.65.1838},
}

@article{mermin1990simple,
  title={Simple unified form for the major no-hidden-variables theorems},
  author={Mermin, N David},
  journal={Physical Review Letters},
  volume={65},
  number={27},
  pages={3373--3376},
  year={1990},
  month = Dec,
  publisher = {American Physical Society},
  doi = {10.1103/PhysRevLett.65.3373},
}

@article{okay2017topological,
  title = {Topological proofs of contextuality in quantum mechanics},
  author = {Okay, Cihan and Roberts, Sam and Bartlett, Stephen D. and Raussendorf,  Robert},
  journal = {Quantum Information and Computation},
  volume = {17},
  number = {13 \& 14},
  pages = {1135--1166},
  year = {2017},
  month = Nov,
  publisher = {Rinton Press},
  doi = {10.26421/QIC17.13-14-5},
}

@article{peres1990incompatible,
  title = {Incompatible results of quantum measurements},
  author = {Peres, Asher},
  journal = {Physics Letters A},
  volume = {151},
  number = {3},
  pages = {107--108},
  year = {1990},
  month = Dec,
  doi = {10.1016/0375-9601(90)90172-K},
}

@article{prevedel2007highspeed,
    author = {Prevedel, Robert and Walther, Philip and Tiefenbacher, Felix and B{\"o}hi, Pascal and Kaltenbaek, Rainer and Jennewein, Thomas and Zeilinger, Anton},
    title = {High-speed linear optics quantum computing using active feed-forward},
    year = {2007},
    month = Jan,
    journal = {Nature},
    volume = {445},
    number = {7123},
    pages = {65--69},
    doi = {10.1038/nature05346},
}

@article{popescu1994quantum,
	author = {Popescu, Sandu and Rohrlich, Daniel},
	title = {Quantum nonlocality as an axiom},
	journal = {Foundations of Physics},
	volume = {24},
	number = {3},
	pages = {379--385},
	year = {1994},
    month = Mar,
	doi = {10.1007/BF02058098},
}

@article{raussendorf:contextualityinmbqc,
  title={Contextuality in measurement-based quantum computation},
  author={Raussendorf, Robert},
  journal={Physical Review A},
  volume={88},
  number={2},
  pages={022322},
  year={2013},
  month = Aug,
  publisher = {American Physical Society},
  doi = {10.1103/PhysRevA.88.022322},
}

@article{raussendorf:cohomologicalframework,
  title = {Cohomological framework for contextual quantum computations},
  author = {Raussendorf,  Robert},
  journal = {Quantum Information and Computation},
  volume = {19},
  number = {13 \& 14},
  pages = {1141--1170},
  ISSN = {1533-7146},
  year = {2019},
  month = Nov,
  publisher = {Rinton Press},
  doi = {10.26421/qic19.13-14-4},
}

@InCollection{raussendorf:paradoxestowork,
  title = {Putting paradoxes to work: Contextuality in measurement-based quantum computation},
  author = {Raussendorf,  Robert},
  booktitle = {{S}amson {A}bramsky on Logic and Structure in Computer Science and Beyond},
  editor = {Palmigiano, Alessandra and Sadrzadeh, Mehrnoosh},
  series = {Outstanding Contributions to Logic},
  volume = {25},
  chapter = {16},
  pages = {595--622},
  year = {2023},
  month = Aug,
  publisher = {Springer},
  doi = {10.1007/978-3-031-24117-8_16},
}

@article{raussendorf2001one,
  title={A one-way quantum computer},
  author={Raussendorf, Robert and Briegel, Hans J},
  journal={Physical Review Letters},
  volume={86},
  number = {22},
  pages = {5188--5191},
  year={2001},
  month = May,
  publisher = {American Physical Society},
  doi = {10.1103/PhysRevLett.86.5188},
}

@article{raussendorf:roleofcohomology,
  title = {The role of cohomology in quantum computation with magic states},
  author = {Raussendorf, Robert and Okay, Cihan and Zurel, Michael and Feldmann, Polina},
  journal = {Quantum},
  volume = {7},
  pages = {979},
  year = {2023},
  month = Apr,
  publisher = {{V}erein zur {F}{\"{o}}rderung des {O}pen {A}ccess {P}ublizierens in den {Q}uantenwissenschaften},
  doi = {10.22331/q-2023-04-13-979},
}

@phdthesis{searle2024dphil,
    author = {Searle, Amy J},
    title = {Correlation and combinatorics: causal contextuality and spin systems},
    year = {2024},
    school = {University of Oxford},
    type = {{DPhil} thesis},
    url = {https://ora.ox.ac.uk/objects/uuid:7f3e70f3-5136-4a64-89e8-858034004181},
}

@misc{searle2024qpl_mapping,
  author       = {Searle, Amy J and Barbosa, Rui Soares and Abramsky, Samson},
  title        = {Mapping temporal correlations to contextuality correlations},
  year         = {2024},
  month        = Jul,
  howpublished = {Talk presented at 21st International Conference on Quantum Physics and Logic (QPL 2024)},
  url          = {https://qpl2024.dc.uba.ar/abstracts/100.pdf},
  nonote       = {Abstract available at \url{https://qpl2024.dc.uba.ar/abstracts/100.pdf}}
}

@misc{slofstra2024operator,
      title={Operator solutions of linear systems and small cancellation}, 
      author={Slofstra, William and Zhang, Lu-Ming},
      year={2024},
      month = Dec,
      eprint={2412.10305},
      archivePrefix={arXiv},
      primaryClass={quant-ph},
      url={https://arxiv.org/abs/2412.10305}, 
      doi = {10.48550/arXiv.2412.10305},
}

@article{supic2020selftesting,
  title = {Self-testing of quantum systems: a review},
  author = {{\ifmmode\check{S}\else\v{S}\fi{}}upi{\ifmmode\acute{c}\else\'{c}\fi{}}, Ivan and Bowles, Joseph},
  journal = {Quantum},
  volume = {4},
  pages = {337},
  year = {2020},
  month = Sep,
  publisher = {{V}erein zur {F}{\"{o}}rderung des {O}pen {A}ccess {P}ublizierens in den {Q}uantenwissenschaften},
  doi = {10.22331/q-2020-09-30-337},
}

@article{walther2005experimental,
    author = {Walther, Philip and Resch, Kevin J and Rudolph, Terry and Schenck, Emmanuel and Weinfurter, Harald and Vedral, Vlatko and Aspelmeyer, Markus and Zeilinger, Anton},
    title = {Experimental one-way quantum computing},
    journal = {Nature},
    volume = {434},
    number = {7030},
    pages = {169--176},
    year = {2005},
    month=Mar,
    publisher={Nature Publishing Group},
    doi = {10.1038/nature03347},
}

@misc{zhao2025scalable,
      title={Scalable tests of quantum contextuality from stabilizer-testing nonlocal games}, 
      author={Zhao, Wanbing and Liew, H W Shawn and Ho, Wen Wei and Liu, Chunxiao and Bulchandani, Vir B},
      year={2025},
      month = Dec,
      eprint={2512.16654},
      archivePrefix={arXiv},
      primaryClass={quant-ph},
      url={https://arxiv.org/abs/2512.16654}, 
      doi={10.48550/arXiv.2512.16654},
}

\end{document}